\documentclass[journal,onecolumn,12pt,article]{IEEEtran}
\usepackage[doublespacing,nodisplayskipstretch]{setspace}
\usepackage{amsmath,graphicx,mathtools}
\usepackage{amssymb,amsthm,verbatim,mathtools,mathabx}
\usepackage[english]{babel}
\usepackage{csquotes}
\usepackage{bm,bbm}
\usepackage{xcolor,tikz}
\usepackage{url,cite}
\usepackage[colorlinks,pdfstartview=FitH,citecolor=blue,linkcolor=blue,urlcolor=blue]{hyperref}
\usepackage{transparent}
\usepackage{mathrsfs}
\usepackage{subcaption} 
\usepackage{graphicx}   
\usetikzlibrary{calc,bending,backgrounds}
\definecolor{lightgray}{RGB}{224,224,224}

\newtheorem{theorem}{Theorem}
\newtheorem{exmp}{Example}
\newtheorem{corollary}{Corollary}
\newtheorem{definition}{Definition}
\newtheorem{proposition}{Proposition}
\newtheorem{lemma}{Lemma}
\newtheorem{remark}{Remark}

\newcommand{\srr}{\mathcal{S}(\mathbf{G})}

\newcommand{\gaussbinom}[2]{\left[ \genfrac{}{}{0pt}{}{#1}{#2} \right]_2}

\author{Hoang~Ly, 
        Emina~Soljanin,
        and V.~Lalitha
\thanks{H.~Ly and E.~Soljanin are with the Department of Electrical and Computer Engineering, Rutgers, the State University of New Jersey, Piscataway, NJ 08854, USA. E-mail: \{\texttt{mh.ly,emina.soljanin}\}@rutgers.edu. V.~Lalitha is with Signal Processing and Communications Research Center, IIIT Hyderabad, India. E-mail: lalitha.v@iiit.ac.in.}
\thanks{This research has been presented in part at the Joint Mathematics Meetings (JMM), Seattle, January 2025, and at the IEEE International Symposium on Information Theory (ISIT), Ann Arbor, June 2025.}}

\begin{document}
\tikzset{every picture/.style={line width=1.15pt}}
\title{On the Service Rate Region of Reed--Muller Codes}

\maketitle



\begin{abstract}
   We study the Service Rate Region of Reed--Muller codes in the context of distributed storage systems. The service rate region is a convex polytope comprising all achievable data access request rates under a given coding scheme. It represents a critical metric for evaluating system efficiency and scalability.  Using the geometric properties of Reed--Muller codes, we characterize recovery sets for data objects, including their existence, uniqueness, and enumeration. This analysis reveals a connection between recovery sets and minimum-weight codewords in the dual Reed--Muller code, providing a framework for identifying those recovery sets. Leveraging these results, we derive explicit and tight bounds on the maximal achievable demand for individual data objects, thereby defining the maximal simplex contained within the service rate region, and the smallest simplex containing it. These two simplices provide a tight approximation to the service-rate region of Reed--Muller codes. 
   
\end{abstract}

\begin{IEEEkeywords}
\noindent
service rate region, Reed--Muller codes, finite geometry, distributed storage, recovery sets, weight enumerator.
\end{IEEEkeywords}

\newpage
\section{Introduction}\label{intro}
Modern computing systems depend on efficient data access from their underlying storage layers to deliver high overall performance. To ensure reliability and balance server load, storage systems commonly replicate data objects across multiple nodes. The level of replication typically reflects the expected demand for each object. \textcolor{black}{However, in modern systems, such as edge computing, access patterns are often skewed and can frequently change \cite{edge:YadgarKAS19, edge:KolosovYMS20}. In such settings, redundancy schemes that incorporate erasure coding are expected to be more effective than replication alone~\cite{SRR:journals/tit/AktasJKKS21}.}

\textcolor{black}{With the advent of distributed storage, various new performance metrics have emerged. 
In addition to the standard high fault-tolerance and low overhead, codes should provide efficient maintenance by, e.g., keeping the required repair bandwidth low or the size of repair groups small~\cite{DSS:journals/ftcit/RamkumarBSVKK22}. The Service Rate Region (SRR) has emerged as a fundamental metric for evaluating the efficiency and scalability of data access in coded distributed storage systems, thereby generalizing the earlier notion of availability \cite{SRR:journals/tit/AktasJKKS21}.}

SRR is defined as the set of all simultaneously supportable data access request rates under a given redundancy scheme. The SRR is a convex \emph{polytope} in $\mathbb{R}^k$, where $k$ is the number of distinct data objects in the system. This polytope offers a precise characterization of the system’s throughput capabilities~\cite{SRR:journals/tit/AktasJKKS21}. Recent work has analyzed the SRR for several families of linear codes. In particular, explicit descriptions of the SRR of binary Simplex codes, first-order Reed–Muller codes~\cite{SRR:conf/isit/KazemiKS20}, and Hamming codes~\cite{Service_Hamming,SRR_Design:lySV2025} have been presented. A notable insight from~\cite{SRR:conf/isit/KazemiKSS20} links the integrality of SRR demand vectors to batch codes, showing a one-to-one correspondence between integral SRRs and batch coding schemes. \textcolor{black}{Authors in~\cite{SRR:journals/tit/AktasJKKS21} investigate the geometric attributes of SRR polytopes, while authors in~\cite{SRR_Design:lySV2025} have demonstrated the theoretical relationship connecting SRR to both combinatorial design theory and majority-logic decoding \cite{MLD:LyS25}}. More broadly, the SRR framework generalizes classical load balancing~\cite{LB:AktasFSW21}, and has been connected to majority-logic decoding, combinatorial design theory, and the \emph{incidence pattern} of minimum-weight dual codewords—that is, how the supports of these codewords intersect across coordinate positions~\cite{SRR_Design:lySV2025}. 
Collectively, these results underscore the utility of linear codes in optimizing data access.

Maximum distance separable (MDS) codes are particularly notable for achieving the optimal trade-off between redundancy and reliability~\cite{CodesForDSS:journals/ACMStorage/ZhirongYKPXYJ25,DSS:journals/pieee/DimakisRWS11}. Their SRRs have been rigorously analyzed under systematic encoding~\cite{SRR:journals/tit/AktasJKKS21}, where a water-filling allocation scheme was shown to be optimal. These results were later refined and generalized in~\cite{SRR:lySV2025}. Furthermore, the authors of~\cite {Service:journals/siaga/AlfaranoKRS24} studied various geometric properties of SRR polytopes of general linear codes, including their volumes. Given these developments, a natural next step is to investigate the SRRs of other code families with rich algebraic structure—most notably higher-order Reed–Muller (RM) codes, whose SRRs remain largely uncharacterized. Although many storage allocation problems can be reformulated as hypergraph matching problems~\cite{allocations:journals/tcom/PengNS21} and solved using linear programming relaxations~\cite{LargeMatching:journals/jct/AlonFHRRS12}, these techniques become challenging to apply when the underlying graph structure associated with data/symbol recovery is highly complex or is not explicitly known. RM codes, despite their elegant algebraic construction, induce intricate, only partially understood geometric structures that significantly complicate SRR analysis~\cite{SRR_RM_ISIT}.

RM codes were introduced in 1954 by Muller~\cite{muller1954application}, and shortly thereafter Reed proposed a majority-logic algorithm~\cite{reed1954class}. Recently, RM codes have gained renewed attention for their capacity-achieving performance over binary symmetric and erasure channels~\cite{kudekar2017reed, Reeves, abbesandon}. Their applications span beyond communication—appearing in 5G NR~\cite{5gnr}, compressed sensing~\cite{calderbank2010construction}, private information retrieval~\cite{beimel2005general}, and quantum computation~\cite{Quantum_RM}. In the context of SRR, their utility comes from the abundance of disjoint recovery sets (also called repair groups) per message symbol, as revealed by Reed’s decoding algorithm, and their rich geometric structure~\cite{rm_survey}.

\textcolor{black}{Our main contribution is a rigorous analysis of the SRR for RM codes of arbitrary parameters. Leveraging the correspondence between RM codes and finite Euclidean geometry, we first characterize the recovery sets for data objects, establishing their existence, uniqueness, and enumeration. This analysis reveals a fundamental link between recovery sets and minimum-weight codewords in the dual code, providing a framework that generalizes Reed's classical majority-logic decoding. Building on these geometric insights, we derive explicit, tight bounds on the maximal achievable demand for individual data objects. To address the complexity of the full SRR polytope, we construct a ``maximal achievable simplex'' (inner bound) and a tight ``enclosing simplex'' (outer bound), proving that the latter is at most a factor of 2 larger than the former. This provides a tight approximation of the service region even where a complete characterization is analytically intractable. We substantiate this intractability by establishing an equivalence between listing all recovery sets and the \emph{coordinate-constrained weight enumerator problem}---enumerating all codewords in the dual RM code constrained to specific values at specific indices. Since the enumerator problem (and related problems such as determining higher-order Hamming weights) remains unsolved for general parameters, enumerating recovery sets is generally impossible. Consequently, characterizing the exact SRR remains an open challenge, making our tight approximation results particularly significant.}

\textbf{Paper Organization.} Section~\ref{sec:Problem_statement} defines the SRR in coded storage systems, introduces the recovery graph, and explains how the SRR problem can be reframed into a graph-theoretic one. Section~\ref{sec:RM_codes} reviews the properties of the Reed-Muller code essential for data recovery analysis. Section~\ref{sec:recovery} presents the main results on recovery sets, including their enumeration and connections to dual codewords. Section~\ref{sec:Service_rate} derives explicit SRR bounds for RM codes using the results established earlier. Section~\ref{sec:Conclusion} concludes the paper.
\subsection*{Nomenclatures and Notations}
This section introduces the notation used throughout the paper. Concepts that are well-known in the literature will be emphasized the first time in \emph{italic}, whereas less standard concepts are formally introduced via \emph{Definition}. Matrices and standard basis vectors are denoted in \textbf{boldface}. $\mathbb{N}, \mathbb{R}$ denote the set of nonnegative integers and real numbers, respectively. The finite field over a prime power \(q\) is denoted as \(\mathbb{F}_q\). A \(q\)-ary linear code \(\mathcal{C}\) with parameters \([n, k, d]_q\) is a \(k\)-dimensional subspace of the \(n\)-dimensional vector space \(\mathbb{F}_q^n\) with minimum Hamming distance $d$. The Hamming weight of a codeword \(\boldsymbol{x}\) in \(\mathcal{C}\) is denoted as \(\text{wt}(\boldsymbol{x})\). The symbols \(\boldsymbol{0}_k\) and \(\boldsymbol{1}_k\) denote the all-zero and all-one column vectors of length \(k\), respectively. The standard basis vector (column) with a 1 at position \(i\) and 0 elsewhere is represented by \(\mathbf{e}_i\). The transpose of a vector $\boldsymbol{v}$ is $\boldsymbol{v}^{\top}$. \(\mathrm{Supp}(\boldsymbol{x})\) denotes the support of a codeword \(\boldsymbol{x}\). The set of positive integers not exceeding \(i\) is denoted as \([i]\). Similarly, \([a, b]\) denotes the set of integers from \(a\) to \(b\), \textcolor{black}{inclusive}, where \(a, b \,\in\, \mathbb{N}\) and \(a < b\). The cardinality of a set $\mathrm{E}$ is denoted by $|\mathrm{E}|$. The symbol \( \subsetneq \) denotes a \emph{proper subset}, which means that for sets \( A \subsetneq B \), every element of \( A \) is also in \( B \), but \( A \ne B \); that is, \( A \) is strictly contained in \( B \). 
\section{Problem Statement}\label{sec:Problem_statement}
\textcolor{black}{We begin by introducing the problem of coded storage, where data objects are linearly encoded across server nodes, and defining a \emph{recovery set} as a collection of servers enabling data retrieval. To characterize the SRR of such systems, we leverage the established connection between SRR and fractional hypergraph matching~\cite{SRR:journals/tit/AktasJKKS21, SRR:lySV2025} by constructing a \emph{recovery hypergraph} in which vertices represent servers and hyperedges correspond to recovery sets. This graph-theoretic reformulation not only facilitates the visualization of recovery structure and overlap but also allows us to employ \emph{fractional matching} and \emph{vertex cover} to derive achievability bounds on sum rates. We conclude the section by outlining the main contributions of this work.}

\subsection{Service Rate of Codes}
Consider a storage system in which $k$ data objects $o_1, \hdots, o_k$ are stored on $n$ servers, labeled $1, \hdots, n$, using an $[n, k]_q$ linear code with generator matrix $\mathbf{G} \in \mathbb{F}^{k\times n}_q$. Let $\boldsymbol{c}_j$ denotes the $j$-th column of $\mathbf{G}$ for $1 \le j \le n$. A \textit{recovery set} for object $o_i$ is a subset of servers whose stored symbols can collectively recover $o_i$. Formally, a set $R \subseteq [n]$ is a recovery set for $o_i$ if $\mathbf{e}_i \in \text{span}(\{\boldsymbol{c}_j : j \in R\})$, which means that the unit vector $\mathbf{e}_i$ lies in the column space of the subsample of $\mathbf{G}$ indexed by $R$. Without loss of generality, we restrict our attention to \emph{minimal} recovery sets—those for which no proper subset suffices to recover the symbol, i.e., \( R \) such that \( \mathbf{e}_i \notin \mathrm{span}(S) \) for all \( S \subsetneq R \). This ensures that no more server resources are used than necessary.

\textcolor{black}{In the specific case of binary codes ($q=2$), we define the \emph{incidence vector} $\boldsymbol{\chi}(R)$ as a length-$n$ binary vector where $\boldsymbol{\chi}(R)_j = 1$ if $j \in R$ and $0$ otherwise. Consequently, a set $R$ serves as a minimal recovery set for an object $o_i$ if and only if $R$ is minimal and:
\[
\mathbf{G} \cdot \boldsymbol{\chi}(R)^{\top} = \mathbf{e}_i.
\]}
Let $\mathcal{R}_i = \{R_{i, 1}, \hdots, R_{i, t_i}\}$ be the $t_i \,\in\, \mathbb{N}$ recovery sets for the object $o_i$. We assume that requests to download object $o_i$ arrive at rate ${\lambda_i}$, for all $i \,\in\, [k]$. We denote the request rates for the object $1, \hdots, k$ by the vector $\boldsymbol{\lambda} = (\lambda_1, \hdots, \lambda_k) \,\in\, \mathbb{R}_{\ge 0}^k$. Let $\mu_l \,\in\, \mathbb{R}_{\ge 0}$ be the average rate at which the server $l\,\in\, [n]$ processes requests for data objects (i.e., \emph{server capacity}). We denote the service rates of servers $1 ,\hdots, n$ by a vector $\boldsymbol{\mu} = (\mu_1,\hdots, \mu_n)$, and assume that the servers have \textit{uniform capacity}, that is, $\mu_j = 1,\, \forall\, j \,\in\, [n]$.

Consider the class of scheduling strategies that assign a fraction of all requests for an object to each of its recovery sets. Let $\lambda_{i,j}$ be the portion of requests for object $o_i$ that are assigned to the recovery set $R_{i,j}, j \,\in\, [t_i]$. The service rate region (SRR) $\mathcal{S}(\mathbf{G}, \boldsymbol{\mu}) \subset \mathbb{R}_{\ge0}^k$ is defined as the set of all request vectors $\boldsymbol{\lambda}$ that can be served by a coded storage system with generator matrix $\mathbf{G}$ and service rate $\boldsymbol{\mu}$. Alternatively, $\mathcal{S}(\mathbf{G}, \boldsymbol{\mu})$ can be defined as the set of all vectors $\boldsymbol{\lambda}$ for which there exist $\lambda_{i,j} \,\in\, \mathbb{R}_{\ge 0}, i \,\in\, [k]$ and $j \,\in\, [t_i]$, satisfying the following constraints:
\begin{align}
        \sum_{j=1}^{t_i}\lambda_{i,\, j} &= \lambda_i, \quad \forall\, i \,\in\, [k], \label{eq:SRR_1}\\    
        \sum_{i=1}^{k}\sum_{\substack{j=1 \\ l \,\in\, R_{i,\, j}}}^{t_i}\lambda_{i,\, j} &\le \mu_l, \quad \forall\, l \,\in\, [n] \label{eq:capacity_constraint}\\
        \lambda_{i,\, j} &\,\in\, \mathbb{R}_{\ge 0}, \quad \forall\, i \,\in\, [k],\ j \,\in\, [t_i] \label{eq:SRR_2}.
\end{align}
Constraints~\eqref{eq:SRR_1} ensure that the demands for all objects are met, while constraints~\eqref{eq:capacity_constraint} guarantee that no server is assigned requests at a rate exceeding its service capacity. Vectors \( \boldsymbol{\lambda} \) satisfying these constraints are called \textit{achievable}. For each component \( \lambda_i \) of an achievable vector \( \boldsymbol{\lambda} \), the associated allocation \( \lambda_{i,j},\ j = 1, \dots, t_i \), represents how the total request \( \lambda_i \) is distributed across recovery sets. The collection of all such achievable vectors $\boldsymbol{\lambda}$ forms the \textit{service polytope} in \( \mathbb{R}_{\ge 0}^k \). An important property of the service polytope is that it is convex, as shown in the following lemma.
\begin{lemma}(\hspace{-0.1mm}\cite{SRR:conf/isit/KazemiKS20}, Lemma 1)\label{lem:convexity} $\mathcal{S}(\mathbf{G}, \boldsymbol{\mu})$ is a non-empty, convex, closed, and bounded subset of $\mathbb{R}_{\ge 0}^k$.
\end{lemma}
Under the uniform-capacity assumption, we may therefore write \( \mathcal{S}(\mathbf{G}, \boldsymbol{1}) \) or simply \( \mathcal{S}(\mathbf{G}) \) in place of \( \mathcal{S}(\mathbf{G}, \boldsymbol{\mu}) \).

\subsection{Recovery Hypergraphs}
A \emph{hypergraph} (or simply a \emph{graph}) is a pair \((V, E)\) in which \(V\) is a finite set of \emph{vertices} and \(E\) is a multiset of subsets of \(V\), called \emph{hyperedges} (or simply \emph{edges}). The \emph{size} of an edge is the cardinality of its defining subset of \(V\). Each generator matrix \(\mathbf{G}\) has a uniquely associated \emph{recovery hypergraph} \(\Gamma_{\mathbf{G}}\) constructed as follows:
\begin{itemize}
  \item There are \(n\) vertices, each corresponding to one distinct column of \(\mathbf{G}\).
\item A hyperedge labeled \( \mathbf{e}_j \) is formed by the set of vertices whose associated columns collectively constitute a recovery set for the basis vector \( \mathbf{e}_j \). If a single column forms a recovery set on its own, we introduce an auxiliary vertex labeled \( \mathbf{0}_k \), and the corresponding edge connects the vertex associated with that column to this auxiliary vertex.

\end{itemize}
For a recovery hypergraph \(\Gamma_{\mathbf{G}}\) and a set of indices \(I = \{j \mid j \,\in\, [k]\} \subseteq [k]\), an \textit{\(I\)-induced subgraph} of \(\Gamma_{\mathbf{G}}\) is obtained by taking only those edges labeled \(\mathbf{e}_j\) for \(j \,\in\, I\), and including all vertices of \(\Gamma_{\mathbf{G}}\) that appear in these edges. An illustrative example of a recovery hypergraph and its $\{3\}$-induced subgraph appear in Example~\ref{ex:recovery_hypergraph}, where \(\mathbf{G}\) is the generator matrix of \(\mathrm{RM}(1,\, 2)\).
\subsection{Fractional Matching and Service Polytopes for Recovery Graphs}
A \emph{fractional matching} in a hypergraph $(V,E)$ is a vector $\boldsymbol{b} \,\in\, \mathbb{R}_{\ge 0}^{|E|}$ whose components $b_\epsilon$, for $\epsilon \,\in\, E$, are nonnegative and satisfy:
\[
\sum_{\epsilon \,\ni\, v} b_\epsilon \;\le\; 1
\quad \text{for each vertex }v \,\in\, V.
\]
\textcolor{black}{This constraint ensures that the total weight assigned to hyperedges incident on any given vertex does not exceed 1. In the context of our recovery hypergraph, this directly corresponds to constraint~\eqref{eq:capacity_constraint}, which enforces that the total request load routed through any server does not exceed its capacity (normalized to \( \mu = 1 \)). In Example~\ref{ex:recovery_hypergraph}, a valid matching $\boldsymbol{b}$ is illustrated on the recovery graph of $\mathbf{G}_{\mathrm{RM(1, 2)}}$ by assigning specific weights to each edge (recovery set).}
    
The set of all fractional matchings in $\Gamma_{\mathbf{G}}=(V,E)$ forms a polytope in $\mathbb{R}_{\ge 0}^{|E|}$, called the \emph{fractional matching polytope}, denoted $\mathrm{FMP}(\Gamma_{\mathbf{G}})$. 

\begin{definition}
Consider a system employing an \([n, k]\) code with generator matrix \( \mathbf{G} \) and uniform server availability, i.e., \( \boldsymbol{\mu} = \boldsymbol{1}_n \). The \emph{service rate} for \( \mathbf{e}_j \) under a specific fractional matching \( \boldsymbol{b} \), denoted \( \lambda_j(\boldsymbol{b}) \), is the sum of the weights \( b_\epsilon \) over all hyperedges \( \epsilon \) labeled by \( \mathbf{e}_j \). The corresponding \emph{service vector} is
\[
\boldsymbol{\lambda}(\boldsymbol{b}) = (\lambda_1(\boldsymbol{b}), \dots, \lambda_k(\boldsymbol{b})),
\]
that is, the vector of service rates for all message symbols \( \mathbf{e}_j \), \( j \in [k] \). Each matching \( \boldsymbol{b} \) defines a valid allocation for \( \boldsymbol{\lambda} \).
\end{definition}
The following result establishes the relationship between achievable service vectors in a storage system and fractional matchings in its associated recovery hypergraph.
\begin{proposition}(\hspace{-0.1mm}\cite{SRR:journals/tit/AktasJKKS21}, Proposition 1)
\( \boldsymbol{\lambda} = (\lambda_1, \dots, \lambda_k) \) is achievable (i.e., satisfying constraints~\eqref{eq:SRR_1}\text{--}\eqref{eq:SRR_2}) if and only if there exists a fractional matching \( \boldsymbol{b} \) in the recovery graph \( \Gamma_{\mathbf{G}} \) such that
\[
\boldsymbol{\lambda} = \boldsymbol{\lambda}(\boldsymbol{b}).
\]
\label{prop:demand_to_matching}
\end{proposition}
By Proposition~\ref{prop:demand_to_matching}, the set of all service vectors $\boldsymbol{\lambda}$ for which corresponding matchings exist in  $\mathrm{FMP}(\Gamma_{\mathbf{G}})$ forms a polytope in $\mathbb{R}_{\ge 0}^k$ known as the service rate region $\srr$. We therefore use the terms service polytope and service rate region interchangeably.


The \emph{size} of a matching $\boldsymbol{b} \in\, \mathbb{R}^{|E|}$ is defined as $\sum\limits_{\epsilon \,\in\, E} b_{\epsilon}$, i.e., the sum of the weights of all its edges. The size of the matching $\boldsymbol{b}$ in Example~\ref{ex:recovery_hypergraph} is $2$. The \emph{matching number} $\nu^*(\Gamma_{\mathbf{G}})$ is the maximum matching size:
\[
\nu^*(\Gamma_{\mathbf{G}}) \;=\; \max_{\boldsymbol{b} \,\in\, \mathrm{FMP}(\Gamma_{\mathbf{G}})}\sum_{\epsilon \,\,\in\,\, E} b_{\epsilon}.
\]
\textcolor{black}{Intuitively, the size of a matching quantifies the total weight distributed across all hyperedges, and the matching number captures the largest such total achievable under vertex capacity constraints.} A \emph{fractional vertex cover} of $(V,E)$ is a vector $\boldsymbol{w} \,\in\, \mathbb{R}^{|V|}$ with nonnegative components $w_v$ such that $\sum_{v \,\in\, \epsilon} w_v \ge 1$ for every edge $\epsilon \,\in\, E$. Its \emph{size} is $\sum\limits_{v \,\in\, V} w_v$. The \emph{vertex cover number} $\tau^*(\Gamma_{\mathbf{G}})$ is the minimum size of any fractional vertex cover:
\[
\tau^*(\Gamma_{\mathbf{G}})
\;=\;
\min_{}
\Bigl\{\,
\sum_{v \,\in\, V} w_v : \boldsymbol{w} \text{ is a fractional vertex cover}
\Bigr\}.
\]
\textcolor{black}{Intuitively, a fractional vertex cover assigns weights to vertices such that every hyperedge accumulates a total weight of at least one.} The size of a cover is the sum of these vertex weights, and the vertex cover number $\tau^*(\Gamma_{\mathbf{G}})$ is the minimum size achievable over all valid covers. The computation of the maximum matching number $\nu^*(\Gamma_{\mathbf{G}})$ is a linear program whose dual problem corresponds to finding $\tau^*(\Gamma_{\mathbf{G}})$. By the theory of linear programming duality~\cite{Linear_programming}, we establish the following properties:
\begin{itemize}
    \item \textbf{Weak Duality:} The size of \emph{any} fractional vertex cover is an upper bound on the size of \emph{any} fractional matching.
    \item \textbf{Strong Duality:} The optimal values coincide, i.e., $\nu^*(\Gamma_{\mathbf{G}}) = \tau^*(\Gamma_{\mathbf{G}})$.
\end{itemize}

Example~\ref{ex:recovery_hypergraph} shows a graph whose both fractional matching number and vertex cover number are $2$.

Framing the SRR problem in terms of graph theory not only reveals its inherent structure but also allows us to leverage established results from the literature. The former is demonstrated in Proposition~\ref{prop:demand_to_matching}, while the latter can be seen from the next result.

\begin{proposition}(\hspace{-0.1mm}\cite{SRR:lySV2025})\label{prop:sum_bound}
For any vector $\boldsymbol{\lambda} = (\lambda_1, \lambda_2, \dots, \lambda_k)$ in the service region $\srr$,
\begin{align}\label{eq:Duality_bound}
\sum_{j=1}^k \lambda_j \;\le\; \nu^*(\Gamma_{\mathbf{G}}) = \tau^*(\Gamma_{\mathbf{G}}).
\end{align}
Moreover, if $I$ is any subset of $[k]$ and $\Gamma'$ is the $I$-induced subgraph of $\Gamma_{\mathbf{G}}$, then also
\[
\sum_{j \,\in\, I} \lambda_j \;\le\ \nu^*(\Gamma') = \tau^*(\Gamma').
\]
\end{proposition}
The proposition above establishes that the size of any (fractional) vertex cover serves as an upper bound on the sum rate of any achievable vector \( \boldsymbol{\lambda} \). We will later show that in scenarios where most hyperedges have large cardinality and exhibit complex overlap—as perfectly exemplified by Reed--Muller codes—this bound provides a simple yet tight estimate of the achievable sum rate. Importantly, it allows us to bypass the complicated task of analyzing edge matchings by working instead with vertex covers. We illustrate all the mentioned concepts in the following example.

\begin{exmp}\label{ex:recovery_hypergraph}
In this example we consider the first-order RM code \( \mathrm{RM}(1, 2) \) of length \( 2^2 = 4 \), which contains 8 codewords of the form:
\(
a_0 \boldsymbol{1} + a_2 \boldsymbol{v}_2 + a_1 \boldsymbol{v}_1, \, \text{where } a_i \,\in\, \mathbb{F}_2,
\)
and the vectors $\boldsymbol{1}, \boldsymbol{v}_2, \boldsymbol{v}_1$ are given as rows of its corresponding \emph{generator matrix} 
\begin{align}
\label{eq:RM12_generator}
\mathbf{G}_{\mathrm{RM}(1,\,2)} = 
\arraycolsep=3pt\def\arraystretch{0.8}
\begin{bmatrix}
\boldsymbol{1} \\
\boldsymbol{v}_2 \\
\boldsymbol{v}_1 \\
\end{bmatrix}
=
\arraycolsep=3pt\def\arraystretch{0.8}
\begin{bmatrix}
1 & 1 & 1 & 1\\
0 & 0 & 1 & 1\\
0 & 1 & 0 & 1\\
\end{bmatrix},
\end{align}
which encodes the message $\boldsymbol{a} = [a_0, a_2, a_1]$ into $\boldsymbol{x} = [x_1, x_2, x_3, x_4] = \boldsymbol{a}\cdot\mathbf{G}_{\mathrm{RM}(1,\,2)}$. In $\mathbf{G}_{\mathrm{RM}(1,\,2)}$, the two rows \(\bigl[\begin{smallmatrix} \boldsymbol{v}_2\\\boldsymbol{v}_1 \end{smallmatrix}\bigr]\) together form all binary vectors of length 2 (there are $2^2 = 4$ such vectors). (A more detailed definition of RM codes is given in Section~\ref{sec:RM_codes}.) We construct the recovery graph $\Gamma$ of the matrix $\mathbf{G}_{\mathrm{RM}(1,\,2)}$ which is given in~\eqref{eq:RM12_generator}. Observe the following linear dependencies among the columns:
\[
\mathbf{e}_1 = c_1 = c_2 + c_3 + c_4, 
\quad
\mathbf{e}_2 = c_1 + c_3 = c_2 + c_4, 
\quad
\mathbf{e}_3 = c_1 + c_2 = c_3 + c_4.
\]
Consequently, the recovery hypergraph is constructed as shown in the left plot of Fig.~\ref{fig:recovery_graph}, where the vertex $i$ corresponds to the column $c_i$ for $i = 1, \dots, 4$. Although the recovery sets for $\mathbf{e}_2$ and $\mathbf{e}_3$ have a uniform size of 2, the recovery sets for $\mathbf{e}_1$ vary in size. Denote as $R$ the recovery set for $\mathbf{e}_1$ of size 3 consisting of vertices $2,\, 3,\, 4$, then $\boldsymbol{\chi}(R) = 0111$. Since column $c_1$ individually forms a recovery set for $\mathbf{e}_1$, we introduce an auxiliary vertex labeled $\mathbf{0}_3$ connected to vertex 1; the resulting edge represents this singleton recovery set, thus avoiding self-loops in the graph.

\textcolor{black}{To determine the service capacity, consider the fractional matching and vertex cover. A valid fractional matching $\boldsymbol{b}$ of size $2$ can be obtained by assigning weights of 0.5 to two specific edges and 0.25 to the remaining four edges ($0.5\cdot2 + 0.25\cdot4 = 2$). Conversely, a valid vertex cover of size 2 is obtained by assigning weights $w_1 = w_4 = 1$ to vertices 1 and 4, and 0 to the others (note that no edge can be formed by vertices $2$ and $3$ only). Since the matching size equals the vertex cover size, we have $\nu^*(\Gamma) = \tau^*(\Gamma) = 2$. By Proposition~\ref{prop:sum_bound}, this implies that any achievable request vector $\boldsymbol{\lambda} = (\lambda_1, \lambda_2, \lambda_3)$ must satisfy $\sum_{i=1}^3 \lambda_i \le 2$.}

The right plot of Fig.~\ref{fig:recovery_graph} illustrates the $\{3\}$-induced subgraph of $\Gamma$. A valid vertex cover of size 2 for this subgraph is easily found by setting $w_1 = w_3 = 1$ and $w_2 = w_4 = 0$. Thus, by Proposition~\ref{prop:sum_bound}, we derive the bound $\lambda_3 \leq 2$. Performing a similar job on the $\{1\}$-induced and $\{2\}$-induced subgraphs of $\Gamma$, we obtain the bounds $\lambda_1 \leq 2$ and $\lambda_2 \leq 2$.
\begin{figure}[ht]
\centering
\begin{tikzpicture}[scale=0.9, font=\small, line cap=round, line join=round, line width = 1.15pt,
                    background rectangle/.style={fill=none},
                    show background rectangle]

  \foreach \i in {1} {
    \node[circle,draw,fill=gray!30,minimum size=0.65cm,inner sep=2pt] (A\i) at (0,0) {\small 1};
  }
  \coordinate (E) at (0,-1.5);

  \coordinate (B) at (3,3);
  \coordinate (C) at (6,0);
  \coordinate (D) at (3,-3);

  \draw[fill=gray!30] (B) circle (9pt) node {\strut 2};
  \draw[fill=gray!30] (C) circle (9pt) node {\strut 3};
  \draw[fill=gray!30] (D) circle (9pt) node {\strut 4};
  \draw[fill=gray!30] (E) circle (9pt) node {\strut $\mathbf{0}_3$};

  \begin{scope}[on background layer]
    \fill[olive!30] (D) -- (B) -- (C) -- cycle;
    \node[color=olive]
      at (barycentric cs:B=1,C=1,D=0.75) {$\mathbf{e}_1\ \tikz[baseline=-0.5ex]\node[fill=black!15, draw=none, inner sep=2pt]{0.25};$};

    \draw[color=orange,thick] (A1) to[bend left=20] node[midway, yshift=7pt, xshift=-15pt]
  {$\tikz[baseline=-0.5ex]\node[fill=black!15, draw=none, inner sep=2pt]{0.25};\ \mathbf{e}_3$} (B);

    \draw[color=orange,thick] (C) to[bend left=20] node[midway, yshift=-5pt, xshift=15pt] {$\mathbf{e}_3\ \tikz[baseline=-0.5ex]\node[fill=black!15, draw=none, inner sep=2pt]{0.25};$}  (D);

    \draw[color=blue,thick] (A1) to[bend right=20] 
    node[
        midway, 
        align=center, 
        yshift=9pt, 
        xshift=-48pt
    ] 
    {
        \tikz[baseline=-0.5ex]\node[fill=black!15, draw=none, inner sep=2pt]{0.25}; \\ 
        $\mathbf{e}_2$
    } (C);
    
    \draw[color=blue,thick] (B) to[bend right=25] node[midway, yshift=28pt, xshift=7pt] {$\mathbf{e}_2\ \tikz[baseline=-0.5ex]\node[fill=black!15, draw=none, inner sep=2pt]{0.5};$} (D);

      \draw[color=olive,thick] (A1) to[] node[midway, yshift=4pt, xshift=-17pt] {$\tikz[baseline=-0.5ex]\node[fill=black!15, draw=none, inner sep=2pt]{0.5}; \mathbf{e}_1$} (E);
  \end{scope}
\end{tikzpicture}
\hspace{2cm} 
\begin{tikzpicture}[scale=0.9, font=\small, line cap=round, line join=round, 
                    background rectangle/.style={fill=none}, 
                    show background rectangle]
  \foreach \i in {1} {
    \node[circle,draw,fill=gray!30,minimum size=0.65cm,inner sep=2pt] (A\i) at (0,0) {\small 1};
  }

  \coordinate (B) at (3,3);
  \coordinate (C) at (6,0);
  \coordinate (D) at (3,-3);

  \draw[fill=gray!30] (B) circle (9pt) node {\strut 2};
  \draw[fill=gray!30] (C) circle (9pt) node {\strut 3};
  \draw[fill=gray!30] (D) circle (9pt) node {\strut 4};

  \begin{scope}[on background layer]
l
    \draw[color=orange,thick] (A1) to[bend left=20] node[midway, yshift=7pt, xshift=-4pt] {$\mathbf{e}_3$} (B);
    \draw[color=orange,thick] (C) to[bend left=20] node[midway, yshift=-5pt, xshift=6pt] {$\mathbf{e}_3$} (D);


      node[midway, above] {$\mathbf{e}_1$} (A1);
  \end{scope}
\end{tikzpicture}

\caption{(Left) The recovery hypergraph of $\mathbf{G}_{\mathrm{RM}(1, 2)}$. The singleton recovery set for $\mathbf{e}_1$ is represented by the edge connecting vertex 1 to the auxiliary vertex $\boldsymbol{0}_3$, while the recovery set of size 3 forms the olive-colored triangle connecting vertices 2, 3, and 4. A fractional matching $\boldsymbol{b}$ is illustrated, with the weight of each edge displayed in the adjacent box. Notably, the total request (sum of incident edge weights) assigned on vertices 1, 2, and 4 reaches the unit capacity of 1, whereas vertex 3 supports a load of 0.75. (Right) The corresponding $\{3\}$-induced subgraph.}
  \label{fig:recovery_graph}
  \end{figure}
\end{exmp}

\subsection{Summary of Results} 
\textcolor{black}{
Our primary objective is to characterize the Service Rate Region $\mathcal{S}(\mathbf{G}, \mathbf{1})$ for Reed--Muller codes $\mathrm{RM}(r, \, m)$ under uniform server capacity. While the SRR for Simplex codes (which are essentially first-order RM codes) and MDS codes are well-understood, such an characterization for SRR of RM codes of general parameters remains unknown. This paper presents a systematic analysis of the SRR for RM codes with arbitrary parameters, establishing the following key results:
\begin{itemize}
    \item \textbf{Geometric Characterization of Recovery Sets (Section \ref{sec:recovery}):} Leveraging the connection between RM codes and finite Euclidean geometry, we characterize the structure of recovery sets for data objects of arbitrary order $\ell$. We prove that for every message symbol, there exists a \emph{unique smallest recovery set} of size $2^\ell$ corresponding to a specific $\ell$-dimensional subspace (Theorems~\ref{thm:RecoverySet} and~\ref{thm:RecoverySetSize}). We further identify and enumerate the ``next-smallest'' recovery sets, showing they correspond to minimum-weight codewords in the dual code (Theorem~\ref{thm:RecoverySetCount}). These results generalize Reed's classical decoding algorithm by providing a framework for the \emph{direct, parallel recovery} of all message symbols, independent of their associated order (Remark~\ref{rm:generalization_Reed}).
     \item \textbf{Connection to Open Combinatorial Problems:} Moreover, we demonstrate that for each message symbol, characterizing all of its recovery sets is equivalent to solving the \emph{coordinate-constrained weight enumerator problem} for the dual RM code (Remark~\ref{rm:Connection_Dual}). As this remains an open problem in coding theory, a complete, exact characterization of the SRR is hard except for special cases (e.g., first-order RM codes). This suggests us to approximate the SRR via its bounding simplices to offer the best characterization of the SRR possible. 
    \item \textbf{Closed-Form Bounds on Maximal Demands (Section \ref{sec:Service_rate}):} Using the geometric properties of these recovery sets, we derive explicit, tight bounds on the \emph{maximal achievable demand} $\lambda_j^{\max}$ for any individual data object $j$ (Theorem~\ref{thm:maximal_achievable_simplex}). We show that objects associated with higher-order message symbols support higher maximal service rates. Additionally, we provide a tight bound on the aggregate service rate for all data objects of the same order (Theorem~\ref{thm:sum_bound_same_order}), and up to a certain order (Theorem~\ref{thm:total_sum_bound}).
    \item \textbf{Tight Approximation via Bounding Simplices (Section \ref{sec:Service_rate}):} Due to the combinatorial complexity of the full SRR polytope, we approximate the region using two bounding simplices. We define the \emph{maximal achievable simplex} $\mathcal{A}$ (an inner bound where all rate vectors are achievable) and construct a minimal \emph{enclosing simplex} $\Omega$ (an outer bound) based on sum-rate constraints (Corollary~\ref{corollary:enclosing_simplex}). A key contribution of this work is proving that the enclosing simplex is \emph{at most a factor of 2 larger} than the achievable simplex, regardless of the code parameters (Corollary~\ref{corollary:enclosing_simplex} and Example~\ref{ex:MDS_RM_compare}). This provides a much tighter approximation of the SRR compared to MDS codes, where the gap can be as large as a factor $k$.
\end{itemize}}

\section{Reed--Muller Codes Preliminaries}\label{sec:RM_codes}
In this section, we introduce Reed--Muller (RM) codes, explore their relationship with Euclidean geometry, and discuss the Reed decoding algorithm of message symbols. This will lay the foundation for characterizing message symbol recovery sets in Section~\ref{sec:recovery}. A table summarizing the important notation is provided at the end of the section.

\subsection{Reed--Muller Codes}
We begin by defining Reed--Muller codes using standard notation from~\cite{Coding:books/MacWilliamsS77}. Let \( t_1, t_2, \dots, t_m \,\in\, \mathbb{F}_2 \) be \( m \) binary variables, and let \( \boldsymbol{t} = (t_1, t_2, \dots, t_m) \) represent the binary \( m \)-tuples (there are \( 2^m \) such tuples in total). Consider a $m$-variate Boolean function \( f(\boldsymbol{t}) = f(t_1, t_2, \dots, t_m) \) that outputs 0 or 1. The vector \( \boldsymbol{f} \) of length \( 2^m \) is derived by evaluating \( f \) for all $2^m$ possible input vectors \( \boldsymbol{t} \).


\begin{definition}
\label{def:RM_code}
The \( r \)-th order (for \( 0 \leq r \leq m \)) binary Reed--Muller code of length \( n = 2^m \), denoted as \( \mathrm{RM}(r, m) \), consists of all vectors \( \boldsymbol{f} \) where \( f(\boldsymbol{t}) \) is a Boolean function that can be expressed as a polynomial of degree at most \( r \). 
\end{definition}

{\color{black} To illustrate, consider the second-order RM code \( \mathrm{RM}(2, 4) \) of length \( 2^4 = 16 \), which contains $2^{11}$ codewords of the form:
\begin{align}\label{eq:message_to_codeword} 
\boldsymbol{x} = \boldsymbol{a} \cdot \mathbf{G}_{\mathrm{RM}(2,4)} = a_0 \boldsymbol{1} + \sum_{i=1}^4 a_i \boldsymbol{v}_i + \sum_{1 \le i < j \le 4} a_{ij} \boldsymbol{v}_i \boldsymbol{v}_j.
\end{align}
where the vectors $\boldsymbol{1}, \{\boldsymbol{v}_i, 1\leq i \leq 4\}$ and $ 
\{\boldsymbol{v}_i\boldsymbol{v}_j, 1\leq i < j\leq 4 \}$, are given as rows of its corresponding \emph{generator matrix} 
\begin{align}\label{eq:RM_generator}
\mathbf{G}_{\mathrm{RM}(2,4)} = 
\arraycolsep=3pt\def\arraystretch{0.8}
\begin{bmatrix}
\boldsymbol{1}\\
\boldsymbol{v}_4\\
\boldsymbol{v}_3\\
\boldsymbol{v}_2\\
\boldsymbol{v}_1\\
\boldsymbol{v}_3\boldsymbol{v}_4\\
\boldsymbol{v}_2\boldsymbol{v}_4\\
\boldsymbol{v}_1\boldsymbol{v}_4\\
\boldsymbol{v}_2\boldsymbol{v}_3\\
\boldsymbol{v}_1\boldsymbol{v}_3\\
\boldsymbol{v}_1\boldsymbol{v}_2
\end{bmatrix}
=
\left[
\begin{array}{cccccccccccccccc}
1 & 1 & 1 & 1 & 1 & 1 & 1 & 1 & 1 & 1 & 1 & 1 & 1 & 1 & 1 & 1\\
0 & 0 & 0 & 0 & 0 & 0 & 0 & 0 & 1 & 1 & 1 & 1 & 1 & 1 & 1 & 1\\
0 & 0 & 0 & 0 & 1 & 1 & 1 & 1 & 0 & 0 & 0 & 0 & 1 & 1 & 1 & 1\\
0 & 0 & 1 & 1 & 0 & 0 & 1 & 1 & 0 & 0 & 1 & 1 & 0 & 0 & 1 & 1\\
0 & 1 & 0 & 1 & 0 & 1 & 0 & 1 & 0 & 1 & 0 & 1 & 0 & 1 & 0 & 1\\
0 & 0 & 0 & 0 & 0 & 0 & 0 & 0 & 0 & 0 & 0 & 0 & 1 & 1 & 1 & 1\\
0 & 0 & 0 & 0 & 0 & 0 & 0 & 0 & 0 & 0 & 1 & 1 & 0 & 0 & 1 & 1\\
0 & 0 & 0 & 0 & 0 & 0 & 0 & 0 & 0 & 1 & 0 & 1 & 0 & 1 & 0 & 1\\
0 & 0 & 0 & 0 & 0 & 0 & 1 & 1 & 0 & 0 & 0 & 0 & 0 & 0 & 1 & 1\\
0 & 0 & 0 & 0 & 0 & 1 & 0 & 1 & 0 & 0 & 0 & 0 & 0 & 1 & 0 & 1\\
0 & 0 & 0 & 1 & 0 & 0 & 0 & 1 & 0 & 0 & 0 & 1 & 0 & 0 & 0 & 1
\end{array}
\right],
\end{align}
} 
where $\boldsymbol{v}_i\boldsymbol{v}_j$ is formed by taking the element-wise product of $\boldsymbol{v}_i$ and $\boldsymbol{v}_j$. In general, the Reed--Muller code \( \mathrm{RM}(r, m) \) is a linear code with length \( n = 2^m \), dimension \( k = \binom{m}{\le r} := \sum\limits_{i=0}^r \binom{m}{i}\), and minimum distance \( d = 2^{m-r} \). It is thus characterized by the code parameters \( [n, k, d] = [2^m, \binom{m}{\le r}, 2^{m-r}] \). When $m \ge r+1$, the dual code of $\mathrm{RM}(r, \,m)$ is $\mathrm{RM}(m-r-1,\, m)$~\cite{Coding:books/MacWilliamsS77}. Throughout this work, we assume $m \ge r+1$, ensuring that the dual code $\mathrm{RM}(m-r-1,\, m)$ is always well-defined.

\subsection{Geometric Interpretation}
Many properties of Reed--Muller codes are elegantly described using finite geometry. In particular, we work within the framework of \textit{Euclidean geometry} \( \mathrm{EG}(m, 2) \) of dimension $m$ over $\mathbb{F}_2$, also known as binary affine geometry of dimension \( m \). This space consists of \( 2^m \) points \( P_i \) for \( i = 1, \dots, 2^m \), each identified with a binary coordinate vector \( \boldsymbol{t} = (t_1, \dots, t_m) \in \mathbb{F}_2^m \). Throughout this paper, we assume the point ordering follows the convention illustrated for \( \mathrm{EG}(4, 2) \) in Table~\ref{table:points}. \textcolor{black}{Structurally, the sequence of points \( P_1, \dots, P_{2^m} \) corresponds to the columns of the matrix formed by vertically stacking the basis vectors \( \boldsymbol{v}_m, \boldsymbol{v}_{m-1}, \dots, \boldsymbol{v}_1 \), where \( \boldsymbol{v}_m \) constitutes the first row and \( \boldsymbol{v}_1 \) the last.} Within this geometry, an \emph{$r$-flat} is defined as an affine subspace of dimension \( r \) containing \( 2^r \) points. Formally, an $r$-flat is a coset expressed as \( \boldsymbol{u} + V \), where \( V \) is an \( r \)-dimensional linear subspace of \( \mathbb{F}_2^m \) and \( \boldsymbol{u} \in \mathbb{F}_2^m \) is a shift vector. While flats need not pass through the origin (the all-zero point), those that do are linear subspaces. The number of such \( r \)-dimensional linear subspaces---or equivalently the cardinality of the \emph{Grassmannian} \( \mathcal{G}_2(m, r) \)---is given by the Gaussian binomial coefficient \( \gaussbinom{m}{r} = \prod_{i=0}^{r-1}\frac{1-2^{m-i}}{1-2^{i+1}} \) (with \( \gaussbinom{m}{0} := 1 \))~\cite{NetworkCode:book/MarcusONA}. These flats generalize the geometric concepts of points, lines, and planes to the finite field setting.

For any subset \( S \) of the points in the Euclidean geometry \( \mathrm{EG}(m, 2) \), we define its \emph{incidence vector} \( \boldsymbol{\chi}(S) \) as a binary vector of length \( 2^m \) with entries:
\[
\boldsymbol{\chi}(S)_i = 
\begin{cases}
1 & \text{if } P_i \,\in\, S, \\
0 & \text{otherwise},
\end{cases}
\]
where \( P_i \) denotes the \( i \)-th point in \( \mathrm{EG}(m, 2) \). This perspective allows us to interpret codewords of \( \mathrm{RM}(r, m) \) as incidence vectors of specific geometric structures within \( \mathrm{EG}(m, 2) \). To illustrate, consider \( \mathrm{EG}(4, 2) \) with points ordered as in Table~\ref{table:points}. The vector \( \boldsymbol{x} = 0000000011111111 \) is a valid codeword in \( \mathrm{RM}(2, 4) \) and corresponds to the subset \( S(\boldsymbol{x}) =  \{P_9, P_{10}, P_{11}, P_{12}, P_{13}, P_{14}, P_{15}, P_{16}\} \). By inspecting the table, we observe that \( S(\boldsymbol{x}) \) consists precisely of the points satisfying the equation \( t_1 = 1 \); thus, it represents an affine flat. In general, the row $\boldsymbol{v}_i$ ($1\le i \le m$) of $\mathbf{G}$ in~\eqref{eq:RM_generator} corresponds to the variable $t_{m+1-i}$, i.e., if $\boldsymbol{v}_i = \boldsymbol{\chi}(S(\boldsymbol{v}_i))$, then $S(\boldsymbol{v}_i)$ is the flat satisfying $t_{m+1-i} = 1$. (See Example~\ref{ex:flat_example}.) Also, the cardinality of a subset \( S(\boldsymbol{x}) \) is simply its Hamming weight \( \text{wt}(\boldsymbol{x}) \). This geometric connection is formalized in the following theorem.

\begin{table*}[t]
\caption{16 Points in $\mathrm{EG}(4, \,2)$.} 
\label{table:points} 
\begin{center}
\begin{tabular}{c c c c c c c c c c c c c c c c c}
   & $P_1$ & $P_2$ & $P_3$ & $P_4$ & $P_5$ & $P_6$ & $P_7$& $P_8$ & $P_9$ & $P_{10}$ & $P_{11}$ & $P_{12}$ & $P_{13}$ & $P_{14}$ & $P_{15}$ & $P_{16}$\\
 $\boldsymbol{v}_4,\,\, \, t_1$ & 0 & 0 & 0 & 0 & 0 & 0 & 0 & 0 & 1 & 1 & 1 & 1 & 1 & 1 & 1 & 1\\
 $\boldsymbol{v}_3, \,\,\,t_2$ & 0 & 0 & 0 & 0 & 1 & 1 & 1 & 1 & 0 & 0 & 0 & 0 & 1 & 1 & 1 & 1 \\
 $\boldsymbol{v}_2, \,\,\,t_3$ & 0 & 0 & 1 & 1 & 0 & 0 & 1 & 1 & 0 & 0 & 1 & 1 & 0 & 0 & 1 & 1 \\
 $\boldsymbol{v}_1, \,\,\,t_4$ & 0 & 1 & 0 & 1 & 0 & 1 & 0 & 1 & 0 & 1 & 0 & 1 & 0 & 1 & 0 & 1 \\
\end{tabular}
\end{center}
\end{table*}
\begin{theorem}\label{thm:MinWeightRM}
(see~\cite[Chapter 13, Theorems 5 \& 8]{Coding:books/MacWilliamsS77})  
\begin{itemize}
    \item Let $f$ be a minimum weight codeword of \( \mathrm{RM}(r, m) \), say $f = \boldsymbol{\chi}(S)$. Then $S$ is an $(m-r)$-dimensional flat in $\mathrm{EG}(m, 2)$.
    \item The minimum-weight codewords of \( \mathrm{RM}(r, m) \) are precisely the incidence vectors of the \((m - r)\)-flats in \( \mathrm{EG}(m, 2) \).
\end{itemize}
\end{theorem}

\begin{theorem}\label{thm:CodewordsRM}
(see~\cite[Chapter 13, Theorem 12]{Coding:books/MacWilliamsS77}) The set of incidence vectors of all \((m - r)\)-flats in \( \mathrm{EG}(m, 2) \) generates the Reed--Muller code \( \mathrm{RM}(r, m) \).
\end{theorem}

We introduce the following lemma about interactions between the flats in \( \mathrm{EG}(m, 2) \), which is necessary in the original proofs of the aforementioned theorems and also in proving other results in this paper.

\begin{lemma}\label{lem:FlatIntersection}
(see~\cite[Chapter 13]{Coding:books/MacWilliamsS77}) Let \( H \) be any flat in \( \mathrm{EG}(m, 2) \) with incidence vector \( \boldsymbol{\chi}(H) \). If \( \boldsymbol{f} \) is the incidence vector of a set \( S \), then the component-wise product \( \boldsymbol{\chi}(H) \cdot \boldsymbol{f} \) yields the incidence vector of the intersection \( S \,\cap\, H \).
\end{lemma}

\begin{exmp} \label{ex:flat_example}
\textcolor{black}{In \( \mathrm{EG}(m, 2) \), each row \( \boldsymbol{v}_i \) of the generator matrix corresponds to the variable \( t_{m+1-i} \). Specifically, \( \boldsymbol{v}_i \) is the incidence vector of the \( (m-1) \)-flat \( t_{m+1-i} = 1 \), while its component-wise complement \( \overline{\boldsymbol{v}}_i \) represents the flat \( t_{m+1-i} = 0 \). Consequently, the product of the first \( m-\ell \) complemented rows,
\[
\prod_{k=\ell+1}^{m} \overline{\boldsymbol{v}}_k = \overline{\boldsymbol{v}}_m \cdot \overline{\boldsymbol{v}}_{m-1} \cdot \ldots \cdot \overline{\boldsymbol{v}}_{\ell+1},
\]
generates the incidence vector of the \( \ell \)-flat defined by the intersection \( (t_1 = 0) \land (t_2 = 0) \land \dots \land (t_{m-\ell} = 0) \), which consists of the points $\{P_1, P_2, \dots, P_{2^{\ell}}\}$.}

The point \( P_{2^m} \), which corresponds to the coordinate vector of all ones (i.e., \( t_1 = \dots = t_m = 1 \)), lies in the intersection flat defined by the product \( \boldsymbol{v}_m \dots \boldsymbol{v}_1 \). For instance, in \( \mathrm{EG}(4, 2) \), the point \( P_{16} \) (where \( t_1 = \dots = t_4 = 1 \)) lies in the intersection flat whose incidence vector is given by the element-wise product \( \boldsymbol{v}_4 \boldsymbol{v}_3 \boldsymbol{v}_2 \boldsymbol{v}_1 = 0000000000000001 \). More details on these standard observations can be found in~\cite{Coding:books/MacWilliamsS77}.
\end{exmp}
\subsection{Reed's Decoding}
One of the earliest and most practical decoding methods for RM codes is the \textbf{Reed decoding algorithm}. While not a maximum-likelihood decoder (i.e., it does not minimize word error probability), it is optimal for minimizing \textit{symbol error probability}, making it particularly well-suited for data recovery tasks where retrieving individual objects is the primary goal.

We illustrate its operation using the \([16, 11, 4]\) second-order code \( \mathrm{RM}(2, 4) \). Its generator matrix is given in~\eqref{eq:RM_generator}. Reed's decoding recovers symbols sequentially, starting with the highest-degree terms. To recover the second-degree symbol \( a_{12} = \boldsymbol{a} \cdot \mathbf{e}_{11} \), we observe that the basis vector \( \mathbf{e}_{11} \) can be decomposed in multiple ways as a sum of columns of \( \mathbf{G} \):
\begin{align*}
\mathbf{e}_{11} 
\, =\, c_1 + c_2 + c_3 + c_4 \, =\, c_5 + c_6 + c_7 + c_8 \,=\, c_9 + c_{10} + c_{11} + c_{12} \,= \,c_{13} + c_{14} + c_{15} + c_{16}.
\end{align*}
Since \( a_{12} = \boldsymbol{a} \cdot \mathbf{e}_{11} \), applying these sums to the codeword yields four independent equations for \( a_{12} \):
\begin{align}\label{eq:recovery_sets}
a_{12} = \sum_{i=1}^4 x_i = \sum_{i=5}^8 x_i = \sum_{i=9}^{12} x_i = \sum_{i=13}^{16} x_i.
\end{align}
These four sets of coordinates--$\{x_1, x_2, x_3, x_4\}$, $\{x_5, x_6, x_7, x_8\}$, $\{x_9, x_{10}, x_{11}, x_{12}\}$, and $\{x_{13}, x_{14}, x_{15}, x_{16}\}$--are called \emph{recovery sets} for \( a_{12} \). For binary RM codes, a recovery set is formally defined as a \emph{minimal} set of codeword symbols whose sum yields the target message symbol.

These disjoint recovery sets provide four independent ``votes" for \( a_{12} \), allowing the correct value to be determined via majority vote even in the presence of errors. Suppose the received vector $\boldsymbol{y}$ differs from the transmitted codeword $\boldsymbol{x}$ due to channel errors, i.e., $\boldsymbol{y} = \boldsymbol{x} + \boldsymbol{e}$, where $\boldsymbol{e} \in \mathbb{F}_2^{16}$ is a binary noise vector. For example, if $y_1,\, y_2,\, y_3$ are flipped (i.e., $y_i = x_i+1, \, i=1,\, 2,\, 3$) while all other symbols are received correctly, then the estimate $\sum_{i=1}^4 y_i$ disagrees with $\sum_{i=5}^8 y_i$, $\sum_{i=9}^{12} y_i$, and $\sum_{i=13}^{16} y_i$. A majority vote among them recovers the correct value of $a_{12}$. This technique is known as \textbf{majority-logic decoding} (MLD). Codes that can be decoded using this scheme are known as \textit{majority-logic decodable} codes. These codes are attractive in practice due to their low decoding complexity and their potential to correct beyond worst-case error bounds~\cite{Coding:books/MacWilliamsS77,Coding:books/PetersonW72,soft_MLD_RM}. For example, in the previous case, even if 8 bits \( x_1, x_2, x_3, x_4, x_5, x_6, x_7, x_8 \) are flipped, the first two votes for \( a_{12} \) remain correct, as $1+1+1+1 = 0$ over $\mathbb{F}_2$. 

\textcolor{black}{Once all second-degree symbols are recovered, their contributions are subtracted from the received codeword, reducing the problem to decoding a first-order RM code. For instance, if we determine via majority vote that \( a_{12} = 1 \), we compute the residual \( \boldsymbol{x}' = \boldsymbol{x} - \boldsymbol{v}_1\boldsymbol{v}_2 \). This operation removes the quadratic term \( v_1 v_2 \) from the polynomial. After repeating this subtraction for all second-degree coefficients ($a_{12}, \dots, a_{34}$), the final residual vector contains only linear and constant terms, effectively becoming a codeword in \( \mathrm{RM}(1, 4) \). This allows the decoder to proceed to recover the first-degree symbols \( a_1, \dots, a_4 \). For a general $\mathrm{RM}(r,m)$ code, this decoding proceeds sequentially in $(r+1)$ stages in a similar way (see an explicit example in Remark~\ref{rm:generalization_Reed}).} 

\begin{table}[h]
\centering
\caption{Summary of Notations for Reed-Muller codes}
\label{tab:notations}
\renewcommand{\arraystretch}{1.2}
\begin{tabular}{c l}
\hline
\textbf{Notation} & \textbf{Description} \\
\hline
$\boldsymbol{t}$ & Binary $m$-tuple $(t_1, \dots, t_m)$ representing coordinates in $\mathrm{EG}(m,2)$. \\
$\boldsymbol{a}$ & Binary message of length $k$ containing symbols $[a_0, a_1, \dots, a_k$]. \\
$\mathbf{G}\in \mathbb{F}_2^{k \times n}$ & Generator matrix of the Reed--Muller code. \\
$\boldsymbol{x}$ & Codeword vector of length $n=2^m$, $\boldsymbol{x} = \boldsymbol{a}\cdot\mathbf{G}$. \\
$c_j$ & The $j$-th column of $\mathbf{G}$ (associated with server $j$). \\
$\boldsymbol{v}_i$ & The row of $\mathbf{G}$ corresponding to the variable $t_{m+1-i}$, i.e., \\
& if $\boldsymbol{v}_i = \boldsymbol{\chi}(S(\boldsymbol{v}_i))$, then $S(\boldsymbol{v}_i)$ is the flat $t_{m+1-i} = 1$. (See also Example~\ref{ex:flat_example}.)\\
$\mathbf{e}_j$ & The $j$-th standard basis (unit) vector of length $k$. \\
$\sigma^{\ell}$ & Length-$\ell$ ordered tuple of distinct indices from $[m]$ identifying a degree-$\ell$ message symbol.\\
$P_i$ & The $i$-th point in the Euclidean geometry $\mathrm{EG}(m, 2)$. \\
$\boldsymbol{\chi}(S)$ & Incidence vector of a subset $S \subseteq \mathrm{EG}(m, 2)$. \\
& Concretely, $\boldsymbol{\chi}(S)_i = 1$ when $P_i \in S$, and $0$ otherwise. \\
\hline
\end{tabular}
\end{table}

\section{Recovery sets of Reed--Muller codes}\label{sec:recovery}
In this section, we characterize the recovery sets for RM code message symbols. We begin by classifying message symbols by their algebraic degree \( \ell \in \{0, \dots, r\} \). For instance, in our example:
\begin{itemize}
    \item \textbf{Order 2:} \( a_{12}, a_{13}, a_{23}, a_{14}, a_{24}, a_{34} \) (corresponding to \( \mathbf{e}_{11} \) down to \( \mathbf{e}_6 \)).
    \item \textbf{Order 1:} \( a_1, a_2, a_3, a_4 \) (corresponding to \( \mathbf{e}_5 \) down to \( \mathbf{e}_2 \)).
    \item \textbf{Order 0:} \( a_0 \) (corresponding to \( \mathbf{e}_1 \)).
\end{itemize}
To formalize this ordering, we define for each \( \ell \in [r] \) a length-\( \ell \) tuple \( \sigma^{\ell} = \sigma_1 \sigma_2 \dots \sigma_{\ell} \), where the indices are drawn without replacement from the set \([m]= \{1, 2, \dots, m\} \) and satisfy \( \sigma_1 < \sigma_2 < \dots < \sigma_{\ell} \). These tuples index the degree-\( \ell \) message symbols. For instance, the tuples \( 12, 13, 14, 23, 24, 34 \) index the second-degree tuples for \( m = 4 \). When \( \ell = 0 \), we define the special tuple \( \sigma^0 = (0) \), and denote \( a_{\sigma^0} := a_0 \).

In our setting of the SRR problem using RM codes, we say a \emph{data object \( j \) (corresponding to vector $\mathbf{e}_j$) is of order \( \ell \)} if its index \( j \) lies in the range:
\begin{equation}\label{eq:object_range}
\sum\limits_{i=0}^{\ell-1} \binom{m}{i} + 1 \le j \le \sum\limits_{i=0}^{\ell} \binom{m}{i},
\end{equation}
i.e., $\ell$ is the smallest number such that $\sum\limits_{i=0}^{\ell} \binom{m}{i} \ge j$. Then there exists a length-\( \ell \) tuple \( \sigma^{\ell} \) such that
\(
a_{\sigma^{\ell}} = \boldsymbol{a} \cdot \mathbf{e}_j,
\)
which a special case $a_{0} = a_{\sigma^0} = \boldsymbol{a}\cdot\mathbf{e}_1$. That is, the unit vector \( \mathbf{e}_j \) corresponds to a symbol of order \( \ell \). In general, there are \( \binom{m}{\ell} \) symbols of order \( \ell \). \textcolor{black}{Conversely, for each \( \ell \,\in\, \{0, 1, \dots, r\} \), define
    \begin{align}\label{eq:data_indices}
    p(\ell) = \sum\limits_{i=0}^{\ell-1} \binom{m}{i} + 1 \quad \text{and} \quad q(\ell) = \sum\limits_{i=0}^{\ell} \binom{m}{i}.
    \end{align}
    Then, $[p(\ell),\, q(\ell)]$ is the range of indices of objects that are associated with message symbols of the same order $\ell$.}

The order of a symbol is critical because, as we will show, the structure of the recovery sets differ across symbols of different orders, directly impacting their achievable service rate in distributed storage systems. The following theorem summarizes Reed’s majority-logic decoding algorithm; see, for example,
\cite[Ch.~13, Thm.~14]{Coding:books/MacWilliamsS77}.

\begin{theorem}[Reed's decoding algorithm]\label{thm:ReedDecoding} Each message symbol of order $r$ (the highest order) \( a_{\sigma^r} \) can be determined by partitioning the \( 2^m \) coordinates of the codeword \( \boldsymbol{x} = \boldsymbol{a} \cdot \mathbf{G} \) into \( 2^{m-r} \) pairwise disjoint subsets of size \( 2^r \), where the sum of the coordinates within each subset equals \( a_{\sigma^r} \). Each such set is a recovery set for $a_{\sigma^r}$.
\end{theorem}
\textcolor{black}{This theorem generalizes the equalities shown in~\eqref{eq:recovery_sets}. Geometrically, the coordinates constituting each recovery set correspond to the points of an \( r \)-flat in \( \mathrm{EG}(m, 2) \). Specifically, if we denote the recovery sets as $S_1, S_2, \dots, S_{2^{m-r}}$, then for each $i$, the point set $\{P_j \mid \, j \in S_i\}$ forms an $r$-flat. The flat containing the origin \( P_1 \) is an \( r \)-dimensional linear subspace, while the remaining \( 2^{m - r} - 1 \) flats are its \emph{affine translations}. For example, in \( \mathrm{RM}(2, 4) \), the four recovery sets in~\eqref{eq:recovery_sets} correspond to the $2$-flats 
\begin{align}
\label{eq:example_RM24}
\{P_1, P_2, P_3, P_4\}, \  \{P_5, P_6, P_7, P_8\},\ \{P_9,  P_{10}, P_{11}, P_{12}\},\ \text{and } \{P_{13},  P_{14}, P_{15}, P_{16}\}.
\end{align}
Among these, the first, which contains the origin $P_1$, is a $2$-dimensional linear subspace, and the others are its affine translations. Note that they are pairwise disjoint.}

While Reed's algorithm effectively utilizes these sets, it is inherently sequential: recovering lower-degree symbols requires decoding and subtracting all higher-degree terms first. To enable efficient, random access in storage systems, we develop a first-ever framework for the \textbf{direct, parallel recovery} of symbols of \emph{arbitrary} order \( \ell \). The following results generalize Theorem~\ref{thm:ReedDecoding}, establishing the existence and structure of recovery sets for any symbol \( a_{\sigma^{\ell}} \):

\begin{itemize}
    \item There exists a unique recovery set $S$ of size \( 2^{\ell} \). The points $\{P_j \mid j \in S\}$ form a specific \( \ell \)-dimensional linear subspace \( \mathscr{S} \subset \mathrm{EG}(m, 2) \) whose construction is explicitly known (Theorem~\ref{thm:RecoverySet}).
    
    \item All other recovery sets for \( a_{\sigma^{\ell}} \) have size at least \( 2^{r+1} - 2^{\ell} \) (Theorem~\ref{thm:RecoverySetSize}). Among these, exactly \( \gaussbinom{m - \ell}{r + 1 - \ell} \) attain the minimum size of \( 2^{r+1} - 2^{\ell} \). Each such recovery set \( T \) corresponds to the complement of \( \mathscr{S} \) within a distinct \( (r+1) \)-dimensional linear subspace containing \( \mathscr{S} \). This establishes a one-to-one correspondence between these subspace complements and the recovery sets. Furthermore, these sets form a \emph{combinatorial $t$-design} with $t=1$: each point \( P_j \in \mathrm{EG}(m, 2) \setminus \mathscr{S} \) is included in exactly \( \gaussbinom{m - \ell - 1}{r - \ell} \) such complements (Theorem~\ref{thm:RecoverySetCount}). (Recall that a \( t\text{--}(n,k,\lambda) \) block design (or a \( t \)-design) is a combinatorial structure consisting of a set \( V \) of \( n \) elements (called \emph{points}) and a collection of $k$-element subsets of \( V \) (called \emph{blocks}), such that every \( t \)-subset of \( V \) is contained exactly in \( \lambda \) blocks.)
\end{itemize}

\begin{theorem}\label{thm:RecoverySet}
For any integer \( \ell \) with \( 1 \leq \ell \leq r \), the symbol \( a_{\sigma^{\ell}} \) can be recovered by summing the codeword coordinates indexed by a specific coordinate subset \( S \subseteq [2^m] \). This set satisfies:
\begin{equation}\label{min_recovery_set}
\begin{cases}
    S \ni 1, \\
    |S| = 2^{\ell}, \\
    \sum_{j \in S} x_j = a_{\sigma^{\ell}}.
\end{cases}
\end{equation}
We refer to \( S \) as a \textit{recovery set} for \( a_{\sigma^{\ell}} \). Geometrically, the points \( \{ P_j \mid j \in S \} \) form an \( \ell \)-dimensional linear subspace \( \mathscr{S} \subset \mathrm{EG}(m, 2) \) (or, \textcolor{black}{an $\ell$-flat passing through the origin $P_1$}).
\end{theorem}

\begin{proof}
    The proof can be found in Appendix~\ref{app:proof_theorem4}.
\end{proof}

With the existence of $S$ established, a natural question arises: is this recovery set unique, or simply one of many? The following result confirms that for $\ell < r$, $S$ is strictly the smallest recovery set, identifying a noticeable ``size gap'' between it and any alternative.

\begin{theorem}\label{thm:RecoverySetSize}
    If \( {\ell} < r \), then \( S \) is the \textbf{unique} recovery set for \( a_{\sigma^{\ell}} \) with size less than \( 2^r \). Any other recovery set must have a size of at least \( 2^{r+1} - |S| = 2^{r+1} - 2^{\ell} \). 
    
    In the case \( {\ell} = r \), the set \( S \) has size \( 2^r \), which matches the size of all other recovery sets for \( a_{\sigma^r} \).
\end{theorem}

\begin{proof}
    The proof can be found in Appendix~\ref{app:proof_theorem5}.
\end{proof}

Having identified the unique smallest recovery set, we now turn to the larger, ``next-best'' recovery sets. The following theorem uses Gaussian binomial coefficients to count these sets and reveals their elegant geometric structure: they are formed by taking larger subspaces and removing the core subspace $\mathscr{S}$.

\begin{theorem}\label{thm:RecoverySetCount}
    For each symbol \( a_{\sigma^{\ell}} \), there are exactly \( \gaussbinom{m - {\ell}}{r + 1 - {\ell}} \) recovery sets of size \( 2^{r+1} - 2^{\ell} \). Each such set corresponds to the complement of \( \mathscr{S} \) within a distinct \( (r+1) \)-dimensional linear subspace of \( \mathrm{EG}(m, 2) \) containing \( \mathscr{S} \) (and is therefore disjoint from $S$).
    
     Furthermore, these recovery sets are distributed evenly across the remaining coordinates: every coordinate \( x_j \) with \( j \notin S \) appears exactly in \( \gaussbinom{m - {\ell} - 1}{r - {\ell}} \) of these sets, forming a 1-design (i.e., a $t$-design with $t=1$) on the set $[n]\setminus S$.
\end{theorem}
\begin{proof}
    The proof can be found in Appendix~\ref{app:proof_theorem6}.
\end{proof}

\begin{remark}
\label{rm:generalization_Reed}
We now show that the findings from Theorems~\ref{thm:RecoverySet}--\ref{thm:RecoverySetCount} together generalize Theorem~\ref{thm:ReedDecoding} which appears as the special case \(\ell = r\). In this specific instance:
\begin{itemize}
    \item The unique smallest recovery set $S$ has size $2^{\ell} = 2^r$. Since $1 \in S$, the set of points $\{P_j \mid j \in S\}$ passes through the origin $P_1$ and forms a linear subspace of $r$ dimensions in $\mathrm{EG}(m, 2)$.
    \item There are \(\gaussbinom{m - {\ell}}{r + 1 - {\ell}} = \gaussbinom{m-r}{1} = 2^{m-r} - 1\) additional recovery sets, each of size $2^{r+1} - 2^{\ell} = 2^{r+1} - 2^r = 2^r$. Geometrically, these correspond to the distinct affine translates of the formed by $\{P_j \mid j \in S\}$.
    \item Each coordinate lies in exactly \(\gaussbinom{m - {\ell} - 1}{r - {\ell}} = \gaussbinom{m-r-1}{0} = 1\) of these recovery sets, confirming that the collection of recovery sets forms a partition of the coordinate indices (i.e., they are pairwise disjoint).
\end{itemize}
\textcolor{black}{Continuing with Example~\ref{ex:recovery_hypergraph} for $\mathrm{RM}(1,\,2)$, classical Reed decoding first characterizes the recovery sets for the message symbols of the highest order ($r=1$):
\begin{align*}
a_1 &= x_1 + x_2 = x_3 + x_4, \\
a_2 &= x_1 + x_3 = x_2 + x_4.
\end{align*}
Upon receiving a possibly corrupted sequence $\boldsymbol{x}' = [x_1', x_2', x_3', x_4']$, the decoder performs a majority vote based on these equations to derive estimates $\hat{a}_1$ and $\hat{a}_2$. To recover the lower-order symbol $a_0$, the contributions of $\hat{a}_1$ and $\hat{a}_2$ must first be subtracted from $\boldsymbol{x}'$. This yields a codeword in $\mathrm{RM}(0,\,2)$, to which Reed's algorithm is applied again:
\begin{align*}
a_0 = x_1 = x_2 - \hat{a}_1 = x_3 - \hat{a}_2 = x_4 - \hat{a}_1 - \hat{a}_2.
\end{align*}
This process highlights the inherent sequentiality of Reed's algorithm: the recovery of $a_0$ (order 0) depends on the successful prior decoding of $a_1$ and $a_2$ (order 1).
In contrast, our results determine the recovery sets for all message symbols simultaneously, independent of their order. Applying our theorems yields:
\begin{align*}
a_1  &= x_1 + x_2 = x_3 + x_4, \\ 
a_2 &= x_1 + x_3 = x_2 + x_4, \\  
a_0  &= x_1 = x_2 + x_3 + x_4.
\end{align*}
Crucially, the equations for $a_0$ enable direct estimation using majority voting on the raw received symbols, without requiring intermediate estimates of $a_1$ and $a_2$. This generalizes Reed's decoding process to enable direct, parallel data recovery for all message symbols.}
\end{remark}

The structural properties established in Theorems~\ref{thm:RecoverySet}--\ref{thm:RecoverySetCount} for message symbols extend directly to the recovery of unit vectors. We explicitly formalize this connection below to facilitate the analysis of the SRR in the next section.

\begin{remark}\label{restate:recovery_set}
For each ${\ell} \in \{0, 1, \dots, r\}$, let \( i \) be any integer satisfying:
    \[
    \sum\limits_{j=0}^{{\ell}-1} \binom{m}{j} + 1 \leq i \leq \sum\limits_{j=0}^{\ell} \binom{m}{j}.
    \]
    In other words, $i$ is the object index associated with a message symbol of order ${\ell}$. Then, for the basis vector \( \mathbf{e}_i \), there exists a coordinate subset \( S \subseteq [2^m] \) of size \( 2^{\ell} \) such that:
    \begin{equation*}
    \begin{cases}
        S \ni 1, \\
        |S| = 2^{\ell}, \\
        \sum\limits_{j \in S} \boldsymbol{c}_j = \mathbf{e}_i.
    \end{cases}
    \end{equation*}
    Furthermore, the recovery set \( S \) satisfies the following uniqueness properties:
    \begin{itemize}
        \item \textbf{When \( {\ell} < r \):} \( S \) is the \textbf{unique} recovery set for \( \mathbf{e}_i \) with cardinality less than \( 2^r \). Any other recovery set for \( \mathbf{e}_i \) must have a size of at least \( 2^{r+1} - |S| = 2^{r+1} - 2^{\ell} > 2^r \). There are exactly \( \gaussbinom{m-{\ell}}{r+1-{\ell}} \) recovery sets of this next-smallest size, and each column \( \boldsymbol{c}_j \) with \( j \in [2^m] \setminus S \) appears in exactly \( \gaussbinom{m-{\ell}-1}{r-{\ell}} \) of them.

        \item \textbf{When \( {\ell} = r \):} The set \( S \) has size \( 2^r \), which matches the size of all other recovery sets for \( \mathbf{e}_i \). There are \( 2^{m-r} \) such recovery sets in total, and they are pairwise disjoint.
    \end{itemize}
\end{remark}
\begin{exmp}\label{ex:RM24}
    Consider the Reed--Muller code \( \mathrm{RM}(2, 4) \) with generator matrix given in~\eqref{eq:RM_generator}. This example aims to identify all recovery sets for the symbol \( a_1 \). First, observe that the vector:
    \[
    \boldsymbol{v}_T = \boldsymbol{v}_4 \boldsymbol{v}_3 \boldsymbol{v}_2 = 0000000000000011,
    \]
    represents the incidence vector of a flat \( T \) consisting of the points \( P_{15} = [1, 1, 1, 0]^{\top}\) and \( P_{16} = [1, 1, 1, 1]^{\top} \). Define the subspace \( \mathscr{S} = \{P_{15} + \boldsymbol{1}_4, P_{16} + \boldsymbol{1}_4\} = \left\{[0, 0, 0, 1]^{\top},\ [0, 0, 0, 0]^{\top}\right\} = \{P_1, P_2\} \), which is a 1-dimensional linear subspace with the incidence vector \( \overline{\boldsymbol{v}}_4 \overline{\boldsymbol{v}}_3 \overline{\boldsymbol{v}}_2 = 1100000000000000\). Therefore, the symbol \( a_1 \) is given by:
    \[
    a_1 = \boldsymbol{a} \cdot \mathbf{e}_5 = x_1 + x_2,
    \]
    indicating that \( S = \{1, 2\} \) is a recovery set for \( a_1 \) with size 2.
    
    Additionally, \( a_1 \) can be expressed in multiple ways as a sum of other coordinates:
    \begin{align*}
        a_1 & = x_1 + x_2 = x_3 + x_4 + x_5 + x_6 + x_7 + x_8 \\ 
        & = x_3 + x_4 + x_9 + x_{10} + x_{11} + x_{12} \quad\quad\,\, = x_5 + x_6 + x_9 + x_{10} + x_{13} + x_{14} \\
            & = x_5 + x_6 + x_{11} + x_{12} + x_{15} + x_{16} \,\,\,\,\quad\,= x_3 + x_4 + x_{13} + x_{14} + x_{15} + x_{16} \\
            & = x_7 + x_8 + x_{9} + x_{10} + x_{15} + x_{16}\,\,\,\,\,\,\,\quad = x_7 + x_8 + x_{11} + x_{12} + x_{13} + x_{14}. 
    \end{align*}
    Equivalently, the standard basis vector \( \mathbf{e}_5 \) (corresponding to symbol $a_1$) can be expressed as:
    \begin{align*}
        \mathbf{e}_5 &= [0, 0, 0, 0, 1, 0, 0, 0, 0, 0, 0]^{\top} = c_1 + c_2 \\
                     & = c_3 + c_4 + c_5 + c_6 + c_7 + c_8 \quad\quad\,\,= c_3 + c_4 + c_9 + c_{10} + c_{11} + c_{12} \\
                     & = c_5 + c_6 + c_9 + c_{10} + c_{13} + c_{14} \quad\,\, = c_5 + c_6 + c_{11} + c_{12} + c_{15} + c_{16} \\
                     & = c_3 + c_4 + c_{13} + c_{14} + c_{15} + c_{16} \,\quad = c_7 + c_8 + c_{9} + c_{10} + c_{15} + c_{16} \\
                     & = c_7 + c_8 + c_{11} + c_{12} + c_{13} + c_{14}.
    \end{align*}
    
    In this case, \( {\ell} = 1 \) and \( S = \{1, 2\} \). According to the proof of Theorem~\ref{thm:RecoverySetCount}, the number of minimum-weight codewords \( \boldsymbol{x} \) in the dual code \( \mathrm{RM}(2, 4)^{\perp} = \mathrm{RM}(1, 4) \) that include \( S = \{1, 2\} \) in their support is:
    \[
    \gaussbinom{4 - 1}{2 + 1 - 1} = \gaussbinom{3}{2} = 7.
    \]
    These codewords, defined via their supports, are:
    \begin{align*}
        \mathrm{Supp}(\boldsymbol{x}^1) & = \{1, 2, 3, 4, 5, 6, 7, 8\};\quad \,\quad\,\,\quad       \mathrm{Supp}(\boldsymbol{x}^2) = \{1, 2, 3, 4, 9, 10, 11, 12\}, \\
        \mathrm{Supp}(\boldsymbol{x}^3) & = \{1, 2, 5, 6, 9, 10, 13, 14\}; \quad\quad
        \mathrm{Supp}(\boldsymbol{x}^4) = \{1, 2, 5, 6, 11, 12, 15, 16\}, \\
        \mathrm{Supp}(\boldsymbol{x}^5) & = \{1, 2, 3, 4, 13, 14, 15, 16\}; \,\,\,\quad
        \mathrm{Supp}(\boldsymbol{x}^6) = \{1, 2, 7, 8, 9, 10, 15, 16\}, \\
        \mathrm{Supp}(\boldsymbol{x}^7) &= \{1, 2, 7, 8, 11, 12, 13, 14\}.
    \end{align*}
    Notably, for each \( j \,\in\, \{3, 4, \dots, 16\} = [16]\setminus S\), the coordinate \( j \) appears in exactly:
    \[
    \gaussbinom{4 - 1 - 1}{2 - 1} = \gaussbinom{2}{1} = 3
    \]
    of these codewords. Figure~\ref{fig:RecSets} illustrates how recovery sets for $\mathbf{e}_5$ and $\mathbf{e}_{11}$ are formed by the columns $\boldsymbol{c}_j$.
\end{exmp}
\textbf{Connecting to Dual Codewords:}
The proof of Theorem~\ref{thm:RecoverySetCount} not only quantifies the number of recovery sets of size $2^{r+1}-2^{\ell}$ for each message symbol but also showcases the relationship between these recovery sets and the minimum-weight codewords in the dual Reed--Muller code. 
We conclude this section by a remark establishing a connection between the set of coordinate-constrained codewords in the dual code and all recovery sets, highlighting the existence of an injection map from the former to the latter. This is followed by Corollary~\ref{Corollary:constrained_minweight_CW}, which demonstrates how these established results can be used to count the number of codewords in RM codes having certain properties.

\begin{remark}[Connection to the coordinate‐constrained enumerator problem]
When \({\ell} < r\), Theorems~\ref{thm:RecoverySet} and~\ref{thm:RecoverySetCount} imply that each message symbol \(a_{\sigma^{\ell}}\) has:
\[
\begin{cases}
\text{1 recovery set } S \text{ of size } 2^{\ell},\\
\gaussbinom{m-{\ell}}{r+1-{\ell}} \text{ recovery sets each of size exactly } 2^{r+1} - 2^{\ell}.
\end{cases}
\]
These are also the smallest recovery sets for \(a_{\sigma^{\ell}}\). A natural question is how to \emph{specify} all other recovery sets for \(a_{\sigma^{\ell}}\), or at least \emph{count} their number. We now show that this question leads to a complex, open problem.

Let \(\Sigma\) be the collection of all recovery sets for \(a_{\sigma^{\ell}}\), and define \(\Pi \subseteq \Sigma\) to be those that are disjoint from \(S\); that is, for every \(R \,\in\, \Pi\), we have \(S \,\cap\, R = \emptyset\). Following the argument in \textbf{Step~3} of the proof of Theorem~\ref{thm:RecoverySetCount}, we establish a one-to-one correspondence between each recovery set \(R \,\in\, \Pi\) and a codeword \(\boldsymbol{x}\) satisfying
\[
\boldsymbol{x} \;\,\in\,\; \mathcal{C}^\perp \;=\; \mathrm{RM}(m-r-1,m),
\quad\text{and}\quad
S \;\subseteq\; \mathrm{Supp}(\boldsymbol{x}).
\]
For instance, in Example~\ref{ex:RM24}, the smallest recovery set for \(a_1\) in \(\mathrm{RM}(2,4)\) is \(S = \{1,2\}\). Determining \emph{all} recovery sets for \(a_1\) in $\Pi$ is then equivalent to identifying all dual codewords \(\boldsymbol{x}\,\in\, \mathcal{C}^\perp\) whose support contains \(\{1,2\}\), i.e., determine $\{\boldsymbol{x} \in \mathrm{RM}(1,4) \mid \{1, 2\} \in \mathrm{Supp}(\boldsymbol{x})\}$.

In a broader sense, identifying all codewords in a Reed--Muller code whose values are constrained to be one at specific coordinate positions is tied to the \emph{coordinate-constrained enumerator} problem. This problem remains unsolved for general-order RM codes, and its resolution is pivotal to fully characterizing their service-rate region. Figure~\ref{fig:connectionWE} illustrates the connection between enumerating all recovery sets, the related enumerator problem, and their relative complexity. The fundamental difficulty of listing or counting these constrained codewords underscores why developing a complete SRR description for RM codes of general parameters remains an open and challenging task. Similarly, determining the sizes of the third-smallest recovery sets is tied to finding the second and higher-order weights of RM codes, a problem that has only been solved in limited cases~\cite{RMWeight:journals/ccds/Rolland10,RMWeight:preprint/arxiv/datta2025}.
\label{rm:Connection_Dual}
\end{remark}

\begin{figure*}
\centering
\begin{tikzpicture}[scale=0.86]
\tikzset{every path/.style={line width=0.6pt}}
\def\tinyspacing{1.03} 
\def\spacing{1.7} 
\def\widespacing{2.8} 
    \foreach \i in {1,...,8} {
        \node[circle,draw,minimum size=0.5cm,,inner sep=2.5pt] (e\i) at ({\i*\spacing+1.1},5) {\small\ifnum\i=1 $S$\fi};
    }

    \foreach \i in {1,...,16} {
        \node[circle,draw,fill=lightgray,minimum size=0.65cm,inner sep=2pt] (s\i) at (\i*\tinyspacing,2) {\small\i};
    }

    \foreach \i in {9,...,12} {
        \node[circle,draw,minimum size=0.5cm,inner sep=2.5pt] (e\i) at ({\i*\widespacing-20.6},-1) {\small\ifnum\i=9 $S$\fi};
    }

    \node[anchor=east] at (0.5,5) {\textcolor{black}{\normalsize{8 recovery sets for $\mathbf{e}_5$}}};
    
    \node[anchor=east] at (0.1,2) {\textcolor{black}{\normalsize{16 columns $\boldsymbol{c}_j$}}};
    
    \node[anchor=east] at (0.5,-1) {\textcolor{black}{\normalsize{4 recovery sets for $\mathbf{e}_{11}$}}};
    
\draw[color=gray][color=gray] (e1) -- (s1) (e1) -- (s2);
\draw[color=gray] (e2) -- (s3) (e2) -- (s4) (e2) -- (s5) (e2) -- (s6) (e2) -- (s7) (e2) -- (s8);
\draw[color=gray] (e3) -- (s3) (e3) -- (s4) (e3) -- (s9) (e3) -- (s10) (e3) -- (s11) (e3) -- (s12);
\draw[color=gray] (e4) -- (s5) (e4) -- (s6) (e4) -- (s9) (e4) -- (s10) (e4) -- (s13) (e4) -- (s14);
\draw[color=gray] (e5) -- (s5) (e5) -- (s6) (e5) -- (s11) (e5) -- (s12) (e5) -- (s15) (e5) -- (s16);
\draw[color=gray] (e6) -- (s3) (e6) -- (s4) (e6) -- (s13) (e6) -- (s14) (e6) -- (s15) (e6) -- (s16);
\draw[color=gray] (e7) -- (s7) (e7) -- (s8) (e7) -- (s9) (e7) -- (s10) (e7) -- (s15) (e7) -- (s16);
\draw[color=gray] (e8) -- (s7) (e8) -- (s8) (e8) -- (s11) (e8) -- (s12) (e8) -- (s13) (e8) -- (s14);

\draw[color=gray] (e9) -- (s1) (e9) -- (s2) (e9) -- (s3) (e9) -- (s4);
\draw[color=gray] (e10) -- (s5) (e10) -- (s6) (e10) -- (s7) (e10) -- (s8);
\draw[color=gray] (e11) -- (s9) (e11) -- (s10) (e11) -- (s11) (e11) -- (s12);
\draw[color=gray] (e12) -- (s13) (e12) -- (s14) (e12) -- (s15) (e12) -- (s16);

\end{tikzpicture}
\caption{Recovery sets for \(\mathbf{e}_5\) and \(\mathbf{e}_{11}\) in $\mathrm{RM}(2,\, 4)$, with edges connecting each column \(\boldsymbol{c}_j,\ j = 1, 2, \dots, 16\) to the recovery sets it belongs to. The sets \(S\) include \(c_1\) and are the smallest recovery sets among all recovery sets for the same unit vector \(\mathbf{e}_j\).}
    \label{fig:RecSets}
\end{figure*}
\begin{figure}
    \centering

\begin{tikzpicture}[node distance=3cm and 2cm, every node/.style={align=center}]
    \node[rectangle, draw, minimum width=2.2cm, minimum height=1cm, fill=gray!10] (A1) {\textsc{Determine $\Sigma$}};
    \node[rectangle, draw, minimum width=2.2cm, minimum height=1cm, fill=gray!10, right=3.7cm] (A2) {\textsc{Determine $\Pi$}};
    \node[rectangle, draw, minimum width=2cm, minimum height=1cm, fill=gray!10, below=1.8cm, xshift=5.16cm] (A3) {\textsc{List all} $\boldsymbol{x} \,\in\, \mathcal{C}^{\perp}$ \textsc{s.t.} $\mathrm{Supp}(\boldsymbol{x})\supseteq S$};

    \draw[->, thick] (A1) -- node[above] {Subsume} (A2);
    \draw[<->, thick] (A2) -- node[right] {Equivalent} (A3);

\end{tikzpicture}
\caption{Connection between specifying $\Sigma$ (the set of all recovery sets for $a_{\sigma^{\ell}}$) and $\Pi$ (the set of recovery sets disjoint from $S$), their relation to the weight enumerator problem, and a comparison of their relative complexity.}
    \label{fig:connectionWE}
\end{figure}

\begin{corollary}\label{Corollary:constrained_minweight_CW}  
We present an interesting result about the codewords of Reed–Muller (RM) codes derived from the theorems above. From the proof of Theorem~\ref{thm:RecoverySet}, we observe that the set of points \( \{P_1, P_2, \dots, P_{2^{\ell}}\} \) forms an \( {\ell} \)-dimensional linear subspace, whose incidence vector is \( \overline{\boldsymbol{v}}_{m} \overline{\boldsymbol{v}}_{m-1} \dots \overline{\boldsymbol{v}}_{m-{\ell}+1} \) (\textcolor{black}{see also the second paragraph in Example~\ref{ex:flat_example}, or the beginning of Example~\ref{ex:RM24}.}). Consequently, the set of indices corresponding to these points, \( S = \{1, 2, \dots, 2^{\ell}\} = [2^{\ell}] \), serves as a recovery set of size \( 2^{\ell} \) for the message symbol \( a_{12\dots {\ell}} \).  

By Theorem~\ref{thm:RecoverySetCount}, the number of minimum-weight dual codewords that include \( S \) in their support is equal to the number of \((r + 1)\)-dimensional linear subspaces that contain the \( {\ell} \)-dimensional linear subspace \( \mathscr{S} \) associated with \( S \). The Gaussian binomial coefficient \( \gaussbinom{m - {\ell}}{r + 1 - {\ell}} \) precisely counts these linear subspaces. Therefore, the number of minimum-weight codewords \( \boldsymbol{x} \) in the dual Reed–Muller code \( \mathrm{RM}(r, m)^\perp = \mathrm{RM}(m - r - 1, m) \) that include the first \( 2^{\ell} \) coordinates in their support is given by \( \gaussbinom{m - {\ell}}{r + 1 - {\ell}} \).  

Equivalently, for each \( {\ell} \,\in\, \{0, 1, \dots, r\} \), the number of minimum-weight codewords \( \boldsymbol{x} \) in the Reed–Muller code \( \mathrm{RM}(r, m) \) that include the first \( 2^{\ell} \) coordinates in their support, \( \{1, 2, \dots, 2^{\ell}\} \subseteq \mathrm{Supp}(\boldsymbol{x}) \), is given by the Gaussian binomial coefficient \( \gaussbinom{m - {\ell}}{m - r - {\ell}} \). For example, in \( \mathrm{RM}(2, 4) \), the number of minimum-weight codewords \( \boldsymbol{x} \) satisfying \( \mathrm{Supp}(\boldsymbol{x}) \supset \{1\} \) is \( \gaussbinom{4}{2} = 35 \); those satisfying \( \mathrm{Supp}(\boldsymbol{x}) \supset \{1, 2\} \) is \( \gaussbinom{3}{1} = 7 \); and those satisfying \( \mathrm{Supp}(\boldsymbol{x}) \supset \{1, 2, 3, 4\} \) is \( \gaussbinom{2}{0} = 1 \).

\end{corollary}  
\section{Service Rate Region of Reed--Muller Codes}\label{sec:Service_rate}
In this section, we leverage the results established in previous sections to analyze the SRR of Reed--Muller codes in distributed storage systems. Specifically, we derive explicit bounds on the maximal achievable demand for individual data objects (formally defined below). These results are grounded in the properties of recovery sets and their connections to the dual code. Additionally, we define the maximal achievable simplex, in which all request rates are achievable, and establish bounds on aggregate rates for data objects associated with message symbols of the same order. These findings have direct implications for the design of efficient and scalable distributed storage systems.

For each \( j \,\in\, [k] = \left\{1, 2, \dots, \sum\limits_{i=0}^r\binom{m}{i}\right\} \), we define the \emph{axis intercept}
\[
\lambda_j^{\text{int}} := \max\{\gamma \,\in\, \mathbb{R} \,|\, \gamma \cdot \mathbf{e}_j \,\in\, \mathcal{S}(\mathbf{G})\},
\]
i.e., the maximum achievable rate for $o_\ell$ when all other demands are zero. Next, for each $j\,\in\,[k]$ define the \emph{coordinate-wise maximum}
\[
   \lambda_{j}^{\max}\; :=\;
   \max\bigl\{\lambda_j \mid \boldsymbol{\lambda}\,\in\,\mathcal S(\mathbf{G})\bigr\},
\]
i.e., the largest request rate for object $o_j$ achievable while other objects may also receive traffic. It was proved in~\cite{SRR:lySV2025} that
\(\lambda_{j}^{\mathrm{int}}=\lambda_{j}^{\max}, \ \forall\, j\,\in\, [k]\). Consequently, dedicating the entire system to object \( o_j \) permits serving at most \( \lambda_{j}^{\max} \) requests. We refer to this quantity as the \emph{maximal achievable demand} (or maximal achievable rate) for object \( j \). Practically, \( \lambda_{j}^{\max} \) represents the highest individual request rate for object \( o_j \) that the system can support in isolation. Characterizing these values is particularly relevant in distributed storage systems, where object demands and popularities are usually skewed~\cite{SRR:journals/tit/AktasJKKS21}.

Define the simplex \( \mathcal{A} \) as:
\(
\mathcal{A} := \text{conv}\Bigl(\bigl\{\boldsymbol{0}_k, \lambda_1^{\text{max}}\mathbf{e}_1, \lambda_2^{\text{max}}\mathbf{e}_2, \dots, \lambda_k^{\text{max}}\mathbf{e}_k\bigr\}\Bigr),
\)
where \( \text{conv}(\mathcal{T}) \) denotes the convex hull of the set \( \mathcal{T} \), defined as \( \mathcal{T} = \{\boldsymbol{v}_1, \dots, \boldsymbol{v}_p\} \subset \mathbb{R}^k \). Specifically, \( \text{conv}(\mathcal{T}) \) consists of all convex combinations of the elements in \( \mathcal{T} \), i.e., all vectors of the form
\[
\sum_{i=1}^p \gamma_i \boldsymbol{v}_i, \quad \text{where } \gamma_i \ge 0 \text{ and } \sum_{i=1}^p \gamma_i = 1,
\]
as described in detail in~\cite{ConvexAnalysis:books/Rockafellar70}.
From Lemma~\ref{lem:convexity}, we know that the service polytopes are convex. Consequently, \( \mathcal{A} \subseteq \mathcal{S}(\mathbf{G}) \). This implies that all points within the simplex \( \mathcal{A} \) are achievable, and we refer to it as the \textit{Maximal achievable simplex}. Practically, it represents the largest simplex fully contained within the SRR, serving as a first approximation of the region. For this reason, its characterization is of significant interest. The following theorem characterizes the maximal achievable demands.

\begin{theorem}\label{thm:maximal_achievable_simplex}
    For each \( j \,\in\, [k] \), 
let $\ell$ be the order of object $j$ determined by~\eqref{eq:object_range}. Then, the maximum achievable demand for \( \mathbf{e}_j \) is 
\begin{align}
\label{eq:max_demand}
    \lambda_j^{\text{max}} = 1 + \dfrac{\gaussbinom{m - \ell}{r - \ell + 1}}{\gaussbinom{m - \ell - 1}{r - \ell}} = 1 + \dfrac{2^m - 2^{\ell}}{2^{r+1} - 2^{\ell}}.
\end{align}
\end{theorem}

\begin{proof}
The proof follows a standard converse-and-achievability argument.
\begin{enumerate}
\item \textbf{Upper Bound (Converse)} \\
Let \( I = \{j\} \), and let \( \Gamma' \) denote the \( I \)-induced subgraph of \( \Gamma_{\mathbf{G}} \). By Theorem~\ref{thm:RecoverySet} and the proof of it, and Theorem~\ref{thm:RecoverySetSize}, the edges in \( \Gamma' \) satisfy the following:
\begin{itemize}
    \item There is one edge \( \epsilon \) of size \( 2^{\ell} \) that contains the vertex node associated with \( \boldsymbol{c}_1 \).
    \item Any other edge \( \eta \ne \epsilon \) has size at least \( 2^{r+1} - 2^{\ell} \). If \( \eta \) has size exactly \( 2^{r+1} - 2^{\ell} \), then \( \eta \,\cap\, \epsilon = \emptyset \). If \( \eta \) has size strictly greater than \( 2^{r+1} - 2^{\ell} \), then \( |\eta \setminus \epsilon| \ge 2^{r+1} - 2^{\ell} \) (by~\eqref{eq:OutsideElements}).
\end{itemize}
Define a fractional vertex cover by assigning weights to the vertex nodes as follows:
\[
w_v =
\begin{cases}
    1, & \text{if } v \text{ is the node associated with } \boldsymbol{c}_1, \\
    0, & \text{if } v \,\in\, \epsilon \text{ and } v \ne \boldsymbol{c}_1, \\
    \dfrac{1}{2^{r+1} - 2^{\ell}}, & \text{if } v \notin \epsilon.
\end{cases}
\]
By the edge-size structure of \( \Gamma' \), this weight assignment satisfies all edge-cover constraints. In other words, this weight assignment gives us a valid vertex cover with size:
\[
\sum_{v} w_v = 1 + (n - 2^{\ell}) \cdot \frac{1}{2^{r+1} - 2^{\ell}} = 1 + \frac{2^m - 2^{\ell}}{2^{r+1} - 2^{\ell}}.
\]
Applying the Proposition~\ref{prop:sum_bound} on the \( I \)-induced subgraph of \( \Gamma_{\mathbf{G}} \), we obtain the desired upper bound.

Alternatively, the same bound can be derived using a \textit{capacity argument}. To minimize total server usage, we consider assigning one unit of demand to the smallest recovery set (size \( 2^{\ell} \)), and the remaining \( \lambda_j - 1 \) units to other recovery sets, each of size at least \( 2^{r+1} - 2^{\ell} \). This yields a total server usage of at least:
   \[
    2^{\ell}\cdot 1 + (\lambda_j - 1)(2^{r+1} - 2^{\ell}).
    \]
    Since the total server capacity is \( n = 2^m \), the feasibility condition implies:
    \[
    2^{\ell} + (\lambda_j - 1)(2^{r+1} - 2^{\ell}) \le 2^m.
    \]
    Solving for \( \lambda_j \), we obtain:
    \[
    \lambda_j \le 1 + \frac{2^m - 2^{\ell}}{2^{r+1} - 2^{\ell}}.
    \]
    Since this bound holds for all achievable values, it applies in particular to \( \lambda_j^{\max} \). This completes the proof of the upper bound.
\item \textbf{Achievability} \\
        To show that the upper bound is achievable, we construct an allocation of demands that attains it.
        
        Let \( R_{j, 1} \) be the unique recovery set of size \( 2^{\ell} \) for \( \mathbf{e}_j \). Denote by \( t = \gaussbinom{m - \ell}{r - \ell + 1} \) the number of additional recovery sets of size exactly \( 2^{r+1} - 2^{\ell} \), which we label as \( R_{j,2}, R_{j,3}, \dots, R_{j,t+1} \). Note that none of them overlap with $R_{j, 1}$. \textcolor{black}{In the special case \( \ell = r \), all recovery sets have the same size \( 2^r \). The set \( R_{j,1} \) is distinguished in that its associated point set contains the origin \( P_{1} \) and forms an \( r \)-dimensional linear subspace \( \mathscr{S} \). The remaining recovery sets correspond to \( r \)-flats in \( \mathrm{EG}(m,2) \), each obtained as an affine translation of \( \mathscr{S} \) in $\mathrm{EG}(m, 2)$. (See the example in~\eqref{eq:example_RM24}.)}
        
        We assign the demands $\lambda_{j, i}$ to recovery set $R_{j, i}$ as follows:
        \[
        \lambda_{j,1} = 1, \quad \lambda_{j,k} = \dfrac{1}{\gaussbinom{m - \ell - 1}{r - \ell}} \quad \text{for } k = 2, 3, \dots, t+1.
        \]
        This assignment ensures that each recovery set of size \( 2^{r+1} - 2^{\ell} \) receives a demand of \( \frac{1}{\gaussbinom{m - \ell - 1}{r - \ell}} \).
        
        According to Remark~\ref{restate:recovery_set}, each node \( x_h \notin S \) is contained in exactly \( \gaussbinom{m - \ell - 1}{r - \ell} \) recovery sets. Therefore, the total demand assigned to any such node is:
        \[
        \sum_{\substack{k=2, \\\text{node $h$ in set $R_{j, k}$}}}^{t+1}
        \lambda_{j,k} = \dfrac{\gaussbinom{m - \ell - 1}{r - \ell}}{\gaussbinom{m - \ell - 1}{r - \ell}} = 1.
        \]
        Thus, we do not use more than 100\% of any node, i.e., the constraint in~\eqref{eq:capacity_constraint} is not violated. Moreover, the total demand serviced by this allocation is:
        \begin{align*}
         \lambda_j = \lambda_{j, 1} + \sum\limits_{k=2}^{t+1}\lambda_{j, k} = 1 + \dfrac{t}{\gaussbinom{m-\ell-1}{r-\ell}} = 1+ \dfrac{\gaussbinom{m-\ell}{r-\ell+1}}{\gaussbinom{m-\ell-1}{r-\ell}}.
        \end{align*}
        Note that 
        \begin{align*}
            \dfrac{\gaussbinom{m-\ell}{r-\ell+1}}{\gaussbinom{m-\ell-1}{r-\ell}} = \dfrac{\prod\limits_{i=0}^{r-\ell}\dfrac{1-2^{m-\ell-i}}{1-2^{i+1}}}{\prod\limits_{i=0}^{r-\ell-1}\dfrac{1-2^{m-\ell-1-i}}{1-2^{i+1}}} = \dfrac{2^m-2^{\ell}}{2^{r+1}-2^{\ell}},
        \end{align*}
         and thus the upper bound in~\eqref{eq:max_demand} is achievable.
    \end{enumerate}
\end{proof}
\begin{remark}\label{rm:increasing_rate}
Note that the Achievability part above can also be established using Theorem~2 in~\cite{SRR_Design:lySV2025}. Also, the theorem above shows that all objects associated with message symbols of the same degree have equal maximal achievable demands. Moreover, for fixed values of $r$ and $m \ge r+1$, observe that $f(\ell) \coloneq \dfrac{2^m-2^{\ell}}{2^{r+1} - 2^{\ell}}$ is an increasing function and $f(0) = \dfrac{2^m-1}{2^{r+1}-1} \ge 2^{m-r-1}$, we have:
    \begin{align*}
    1 + 2^{m-r-1} \le \lambda_1^{\text{max}} 
    \le \lambda_2^{\text{max}} = \lambda_3^{\text{max}} = \dots 
    = \lambda_{m+1}^{\text{max}} \le \lambda_{m+2}^{\text{max}} = \dots 
    \le \lambda_k^{\text{max}} = 2^{m-r} = d_{\mathrm{RM}(r, m)}.
\end{align*}
    We observe that objects associated with higher-order symbols have a greater maximal achievable demand, meaning they can be supported at higher demand. This implies that these objects exhibit higher availability, \textcolor{black}{suggesting that more frequently accessed or critical data should be encoded as higher-order RM symbols}.
\end{remark}

From Remark~\ref{rm:increasing_rate}, we see that in the SRR of $r$-th order RM code,
\[
1+2^{m-r-1} \le \lambda_j^{\text{max}} \le 2^{m-r}, \quad \forall\, j.
\]
Comparing the SRR of the $r$-th order RM code with parameters $[2^m, \sum_{j=0}^r \binom{m}{j}, 2^{m-r}]$ and the Simplex code (which is essentially a first-order $\mathrm{RM}(1,\, m)$ code) with approximately the same number of servers, having parameters $[2^m - 1, m, 2^{m-1}]$, we recall that in Simplex code, maximal achievable demand for each individual data object is~\cite{SRR:conf/isit/KazemiKS20,SRR:journals/tit/AktasJKKS21}:
\[
\lambda_j^{\text{max}} = 2^{m-1}, \quad \forall\, j.
\]
Clearly, when $r > 1$, we have $2^{m-1} > 2^{m-r}$, indicating that the maximal achievable demand for each data object in systems using the Simplex code is larger than that in systems using $r$-th order RM codes with (approximately) a similar number of servers. However, this improvement is achieved at the expense of fewer encoded data objects in storage ($m$ vs. $\sum_{j=0}^r \binom{m}{j}$), i.e., a lower code rate.
\textcolor{black}{From a system design perspective, the order \( r \) allows us to balance the trade-off between service availability and storage efficiency. When the system (with fixed number of servers $n$) needs to support very high demand for specific data objects---often referred to as popular content---Simplex codes (or RM codes with \( r=1 \)) are the optimal choice. Their high maximal service rate of \( 2^{m-1} \) ensures they can handle heavy traffic, justifying the fact that they store fewer data objects. Conversely, when the priority is to store as many data objects as possible on a fixed number of servers, higher-order RM codes (\( r > 1 \)) are preferable. They significantly increase the number of stored objects, \( k \), while maintaining robust service capacity. Therefore, a practical resource allocation strategy is to use low-order RM codes for a small number of popular objects to maximize individual service rates for these objects, and higher-order RM codes for less frequently accessed objects to optimize storage efficiency (i.e., the number of encoded data objects).}

    For each data symbol order \( \ell \,\in\, \{0, 1, \dots, r\} \) defined as in~\eqref{eq:data_indices}, we have the following bound that holds for the sum of demands \( \lambda_j \) of such objects.
\begin{theorem}\label{thm:sum_bound_same_order} In an $\mathrm{RM}(r,\,m)$-coded storage system, the total request rate for all objects associated with message symbols of the same order $\ell$, for each $\ell \,\in\, \{0, 1, \dots, r\}$, is upper bounded by:
    \begin{align*}
    \sum\limits_{j = p(\ell)}^{q(\ell)} \lambda_j & \leq 1 + \dfrac{2^m - 2^\ell}{2^{r+1} - 2^\ell}\\
    & = \lambda_i^{\max}, \quad \forall\, i \in [p(\ell),\, q(\ell)].
    \end{align*}
    Moreover, this bound is tight, meaning that there exists an allocation of demands \( \{\lambda_j\} \) for which equality is achieved.
\end{theorem}

\begin{proof}
    The proof can be found in Appendix~\ref{app:proof_theorem8}.
\end{proof}

\begin{remark}
    Theorem~\ref{thm:sum_bound_same_order} implies the followings:
    \begin{itemize}
    \item \textbf{Aggregate Bound for Order-\( \ell \) Symbols:} \\
    The total service rate across all data objects of order \( \ell \), i.e., those indexed from \( p(\ell) \) to \( q(\ell) \), cannot exceed the individual maximum achievable demand \( \lambda_i^{\max} \) for any such object. This reflects a global constraint on how much cumulative demand the system can support for symbols of the same order:
    \[
    \sum_{j = p(\ell)}^{q(\ell)} \lambda_j \le \lambda_i^{\max}, \quad \forall\, i \in [p(\ell),\, q(\ell)].
    \]

    \item \textbf{Shared Resource Limitation:} \\
    Although each object of order \( \ell \) individually can be served up to \( \lambda_i^{\max} \), the system lacks sufficient capacity to serve all such objects at that maximum simultaneously. The inequality implies a \emph{resource-sharing} limitation among objects of the same order.
\end{itemize}
\end{remark}
\begin{theorem}\label{thm:total_sum_bound}
 For each \( {\ell} \,\in\, \{0, 1, \dots, r\} \), the total service rate of all data objects of order at most \( \ell \) is upper bounded by:
\begin{align}
\sum\limits_{j=1}^{q({\ell})} \lambda_j\le 1+\dfrac{2^m -1}{2^{r+1}-2^{\ell}}.
\end{align}
\end{theorem}

\begin{proof}
    Let \( I = \{j \mid 1\le j \le q({\ell})\} \) and let \(\Gamma'\) be the \(I\)-induced subgraph of \(\Gamma_{\mathbf{G}}\). By Theorems~\ref{thm:RecoverySet} and~\ref{thm:RecoverySetSize}, the edges in \(\Gamma'\) satisfy one of the following conditions:
    \begin{itemize}
        \item They have size \emph{at least} \(2^{r+1} - 2^{\ell}\), or
        \item They have size \emph{at most} \(2^{\ell}\) and contain the vertex node associated with \(\boldsymbol{c}_1\).
    \end{itemize}
    Assign weights to vertices as follows:
    \[
        w_1 = 1,\,\, \text{and}\quad w_j = 1/(2^{r+1} - 2^{\ell}),\, \forall\, j > 1.
    \]
    From the edge size properties of \(\Gamma'\), this assignment forms a valid fractional vertex cover whose size is:
    \[
    \sum\limits_{j=1}^n w_j = 1 + (n-1)\cdot\dfrac{1}{2^{r+1} - 2^{\ell}} = 1 + \dfrac{2^m-1}{2^{r+1} - 2^{\ell}} 
    \]
Applying the Proposition~\ref{prop:sum_bound} on the \( I \)-induced subgraph of \( \Gamma_{\mathbf{G}} \), we obtain the desired upper bound.
\end{proof}
\begin{corollary}\label{corollary:enclosing_simplex}
When ${\ell} = r$, Theorem~\ref{thm:total_sum_bound} provides the following neat bound on the heterogeneous sum rate:
\begin{align} 
    \sum\limits_{j=1}^{k} \lambda_j \le 1 + \frac{2^m - 1}{2^r} < 1 + 2^{m-r}.
\end{align}
Thus, the inequality
\[
\sum_{j=1}^{k} \lambda_j \le 1 + 2^{m-r}
\]
defines a simplex $\Omega := \text{conv}\Bigl(\bigl\{\boldsymbol{0}_k, h\mathbf{e}_1, h\mathbf{e}_2, \dots, h\mathbf{e}_k\bigr\}\Bigr)$ that strictly encloses the service region, where $h = 1+2^{m-r}$.
\end{corollary}
Recalling from Remark~\ref{rm:increasing_rate} that
\[
1+ 2^{m-r-1} \le \lambda_j^{\text{max}} \le 2^{m-r}, \quad \forall\, j,
\]
and noting that
\[2 > \frac{1 + 2^{m-r}}{1+ 2^{m-r-1}} > \frac{1 + 2^{m-r}}{2^{m-r}} \approx 1,\]
we conclude from Corollary~\ref{corollary:enclosing_simplex} that the enclosing simplex $\Omega$ is at most a factor of 2 larger than the maximal achievable simplex $\mathcal{A}$, regardless of the RM code parameters. This bound of 2 is significantly tighter than the worst-case factor-$k$ gap observed in systematic MDS codes~\cite{SRR:lySV2025}. In other words, for RM codes, the two simplices provide a much closer approximation of the SRR. \textcolor{black}{Example~\ref{ex:MDS_RM_compare} compares the two bounding simplices of the SRR of $\mathrm{RM}(1,\, 3)$ with those of two MDS codes sharing the same parameters $(n, k)$--one systematic and the other non-systematic.}

\textcolor{black}{Nevertheless, comparing the SRR across different code families is not straightforward, as the criteria for a fair comparison remain open to debate. It is unclear, for instance, whether one should compare codes of the same blocklength $n$, the same dimension $k$, or both. Even for codes sharing identical parameters, the choice of encoding---such as systematic versus non-systematic forms for MDS codes---yields significantly different service regions. Furthermore, differences in code alphabets complicate the analysis: while MDS codes are typically defined over sufficiently large fields, the RM codes investigated here are binary. This distinction leads to disparities in decoding complexity and implementation costs, factors which must also be considered for a holistic comparison.}

\begin{exmp}
\textcolor{black}{Fig.~\ref{fig:SRR_RM12} depicts the maximal achievable simplex $\mathcal{A}$ and minimal enclosing simplex $\Omega$ for the $\mathrm{RM}(1, 3)$ code projected onto the $(\lambda_1, \lambda_2, \lambda_3)$ subspace. These are shown alongside the corresponding simplices for systematic and non-systematic MDS codes of the same code parameters ($n=8, k=4$). Note that since the full SRRs reside in $\mathbb{R}^4$, their visualization requires projection onto a 3-dimensional subspace.
\begin{itemize}
    \item Fig.~\ref{fig:left}: The service region of $\mathrm{RM}(1,\,3)$ is sandwiched between $\mathcal{A}$---which intersects the axes of the subspace $(\lambda_1, \lambda_2, \lambda_3)$ at $\left(\dfrac{10}{3}, 4, 4\right)$---and $\Omega$, defined by $\sum_{j=1}^4\lambda_j \le 4.5$.
    \item Fig.~\ref{fig:mid}: The service region of a non-systematic $(8, 4)$ MDS code coincides exactly with both $\mathcal{A}$ and $\Omega$, as both are described by $\sum_{j=1}^4\lambda_j \le 2$. While this region is perfectly characterized by its bounding simplices, it yields a strictly smaller SRR than $\mathrm{RM}(1, 3)$.
    \item Fig.~\ref{fig:right}: In contrast, the service region of a systematic $(8, 4)$ MDS code is strictly sandwiched between $\mathcal{A}$ $\left(\sum_{j=1}^4\lambda_j \le \dfrac{11}{4}\right)$ and $\Omega$ \big($\sum_{j=1}^4\lambda_j \le 5$\big)~\cite{SRR:lySV2025}. Consequently, the region is less tightly characterized by these simplices than in the previous cases. Furthermore, the actual SRR of the systematic MDS code is not nested within that of $\mathrm{RM}(1, 3)$, nor vice versa. For instance, the request vector $(\lambda_1, \lambda_2, \lambda_3, \lambda_4) = (1, 2, 1, 1)$ is achievable by the systematic MDS code but lies outside the service region of $\mathrm{RM}(1, 3)$, while the vector $(0, 4, 0, 0)$ is achievable by the latter but not the former.
\end{itemize}
}
\label{ex:MDS_RM_compare}
    \begin{figure}
    \centering
    \begin{subfigure}[b]{0.33\linewidth} 
        \centering
        \resizebox{\linewidth}{!}{
\begin{tikzpicture}
    [scale=1.25,
    axis/.style={->, color=black, thick},
    facet/.style={fill=gray!, opacity=0.200000},
    edge/.style={color=black, thick},
    smallfacet/.style={fill=red, opacity=0.2},
    smalledge/.style={color=red!50!black, thick},
    vertex/.style={inner sep=1.5pt,circle,fill=black}]

    \coordinate (O) at (0,0); 

    \coordinate (A) at (90:4.5);   
    \coordinate (B) at (210:4.5);  
    \coordinate (C) at (330:4.5);  

    \coordinate (SA) at (90:3.33);
    \coordinate (SB) at (210:4);
    \coordinate (SC) at (330:4);

    \draw[axis] (O) -- (90:5) node[above] {\huge $\lambda_1$};
    \draw[axis] (O) -- (210:5) node[left] {\huge $\lambda_3$};
    \draw[axis] (O) -- (330:5) node[right] {\huge $\lambda_2$};
    
    \node[vertex, label={below:\huge $0$}] at (O) {};

    \fill[smallfacet] (SA) -- (SB) -- (SC) -- cycle;
    \draw[smalledge] (SA) -- (SB) -- (SC) -- cycle;
    
    \fill[facet] (A) -- (B) -- (C) -- cycle;
    \draw[edge] (A) -- (B) -- (C) -- cycle;
    \draw[-,thick] (1.0,2.3) -- (2.8,3.0) node[right, fill=gray!20, inner sep=5pt] {\Huge $\Omega$};
    \draw[-,thick] (0.6,0.7) -- (2.8,1.7) node[right, fill=red!25, inner sep=5pt] {\Huge $\mathcal{A}$};
    
    \node[above, text=red!60!black] at (SA) {\huge $10/3$};
    \node[above left, text=red!60!black] at (SB) {\Huge $4$};
    \node[above right, text=red!60!black] at (SC) {\Huge $4$};

    \node[right, text=black] at (A) {\huge $4.5$};
    \node[below right, text=black] at (B) {\huge $4.5$};
    \node[below left, text=black] at (C) {\huge $4.5$};

\end{tikzpicture}
}
        \caption{$\mathcal{A},\, \Omega$, and $\mathcal{S}(\mathbf{G})$ of $\mathrm{RM}(1,\, 3)$.}
        \label{fig:left}
    \end{subfigure}
    \hfill 
    \begin{subfigure}[b]{0.28\linewidth}
        \centering
        \resizebox{\linewidth}{!}{    
\begin{tikzpicture}
    [scale=1.2,
    axis/.style={->, color=black, thick},
    facet/.style={fill=gray!, opacity=0.200000},
    edge/.style={color=black, thick},
    smallfacet/.style={fill=red, opacity=0.2},
    smalledge/.style={color=red!50!black, thick},
    vertex/.style={inner sep=1.5pt,circle,fill=black}]

    \coordinate (O) at (0,0); 

    \coordinate (A) at (90:2);   
    \coordinate (B) at (210:2);  
    \coordinate (C) at (330:2);  


    \draw[axis] (O) -- (90:2.5) node[above] {\Large $\lambda_1$};
    \draw[axis] (O) -- (210:2.5) node[left] {\Large $\lambda_3$};
    \draw[axis] (O) -- (330:2.5) node[right] {\Large $\lambda_2$};
    
    \node[vertex, label={below:\Large $0$}] at (O) {};

    
    \fill[facet] (A) -- (B) -- (C) -- cycle;
    \draw[edge] (A) -- (B) -- (C) -- cycle;
    \draw[-,thick] (0.7,0.3) -- (2.2,1.0) node[right, fill=gray!20, inner sep=5pt] {\Large $\Omega \equiv \mathcal{A} \equiv \mathcal{S}(\mathbf{G})$};
    

    \node[right, text=black] at (A) {\Large $2$};
    \node[below right, text=black] at (B) {\Large $2$};
    \node[below, text=black] at (C) {\Large $2$};

\end{tikzpicture}
 }
        \caption{$\mathcal{A},\, \Omega$, and $\mathcal{S}(\mathbf{G})$ of non-systematic MDS $(8, 4)$.}
        \label{fig:mid}
    \end{subfigure}
    \hfill 
   \begin{subfigure}[b]{0.37\linewidth}
        \centering
        \resizebox{\linewidth}{!}{    
\begin{tikzpicture}
    [scale=1.3,
    axis/.style={->, color=black, thick},
    facet/.style={fill=gray!, opacity=0.200000},
    edge/.style={color=black, thick},
    smallfacet/.style={fill=red, opacity=0.2},
    smalledge/.style={color=red!50!black, thick},
    vertex/.style={inner sep=1.5pt,circle,fill=black}]

    \coordinate (O) at (0,0); 

    \coordinate (A) at (90:5);   
    \coordinate (B) at (210:5);  
    \coordinate (C) at (330:5);  

    \coordinate (SA) at (90:2.75);
    \coordinate (SB) at (210:2.75);
    \coordinate (SC) at (330:2.75);

    \draw[axis] (O) -- (90:5.5) node[above] {\Large $\lambda_1$};
    \draw[axis] (O) -- (210:5.5) node[left] {\Large $\lambda_3$};
    \draw[axis] (O) -- (330:5.5) node[right] {\Large $\lambda_2$};
    
    \node[vertex, label={below:\Large $0$}] at (O) {};

    \fill[smallfacet] (SA) -- (SB) -- (SC) -- cycle;
    \draw[smalledge] (SA) -- (SB) -- (SC) -- cycle;
    
    \fill[facet] (A) -- (B) -- (C) -- cycle;
    \draw[edge] (A) -- (B) -- (C) -- cycle;
    \draw[-,thick] (1.3,2.3) -- (2.8,3.0) node[right, fill=gray!20, inner sep=6pt] {\Huge $\Omega$};
    \draw[-,thick] (0.6,0.7) -- (2.8,1.7) node[right, fill=red!25, inner sep=6pt] {\Huge $\mathcal{A}$};
    
    \node[right, text=red!60!black] at (SA) {\Large $11/4$};
    \node[below right, text=red!60!black] at (SB) {\Large $11/4$};
    \node[below, text=red!60!black] at (SC) {\Large $11/4$};

    \node[right, text=black] at (A) {\huge $5$};
    \node[below right, text=black] at (B) {\huge $5$};
    \node[below, text=black] at (C) {\huge $5$};

\end{tikzpicture}
 }
        \caption{$\mathcal{A},\, \Omega$ of systematic MDS $(8, 4)$.}
        \label{fig:right}
    \end{subfigure}    
\caption{(a) The service region of the first-order code $\mathrm{RM}(1,\,3)$ is bounded by its maximal achievable simplex $\mathcal{A}$ (in red) and its minimal enclosing simplex $\Omega$ (in gray). (b) The service region of a non-systematic $(8, 4)$ MDS code is defined by $\sum_{j=1}^4\lambda_j \le 2$, which coincides with its two bounding simplices. (c) The service region of a systematic $(8, 4)$ MDS code is strictly sandwiched between $\mathcal{A}$ \big($\sum_{j=1}^4\lambda_j \le \frac{11}{4}$\big) and $\Omega$ \big($\sum_{j=1}^4\lambda_j \le 5$\big)~\cite{SRR:lySV2025}.}
    \label{fig:SRR_RM12}
\end{figure}
\end{exmp}


\section{Conclusions and Future Work}\label{sec:Conclusion}
\textcolor{black}{This paper presents a detailed analysis of the Service Rate Region SRR of Reed--Muller codes, highlighting their utility in distributed storage systems. By leveraging finite geometry, we characterized the structure and enumerated the recovery sets, which in turn served as the basis for deriving tight bounds on the maximal achievable request rates. A key contribution of this work is the approximation of the SRR via bounding simplices, which we proved are within a factor of 2 of each other. These insights provide practical guidelines for optimizing data accessibility and resource allocation in scalable systems. However, determine which coding scheme will win in practice requires consideration and evaluation at the system level, as explained in \cite{edge:KolosovASY23}.}

\textcolor{black}{Characterizing the recovery sets was instrumental in deriving the bounding simplices of the RM SRR polytopes. Furthermore, this characterization was later used to understand the relationship between the SRR problem and combinatorial design theory \cite{SRR_Design:lySV2025}, and to derive new MLD algorithms for RM codes \cite{MLD:LyS25}.
Future research directions include extending this analysis to $q$-ary Reed--Muller codes ($q > 2$), as they may offer higher code rates (i.e., higher storage density) than their binary counterparts.}

\textcolor{black}{We emphasize again that SRR is one of the several performance metrics of distributed storage systems. It quantifies the efficiency and scalability of the data access. These attributes are particularly important in systems with highly variable data access, such as Edge \cite{edge:KolosovYMS20}. The mean download time is a related performance metric \cite{Download:KadheSS15, Download:AktasKSS21}. 
However, in coded distributed storage systems, evaluating download time is often intractable because codes induce complex queues. Even in simple cases, one has to resort to heuristics \cite{Download:AktasS18}. SRR  represents the stability region of these queues. Distributed systems often require optimizing multiple metrics, such as storage efficiency, reliability, and ease of repair. Thus, characterizing trade-offs between them is an important direction for future research.}

\section*{Acknowledgment}
Part of this work was done during the visit of the last two authors’ to IISc Bengaluru. The authors thank 
the hosts for their hospitality. This work was supported in part by NSF CCF-2122400. The work of V.~Lalitha was supported in part by the grant DST/INT/RUS/RSF/P-41/2021
from the Department of Science \& Technology, Govt.\ of India.\ The authors thank Michael Schleppy for valuable discussions.

\appendices
\section{Proof of Theorem~\ref{thm:RecoverySet}}
\label{app:proof_theorem4}
\begin{proof}
Consider the code's generator matrix \( \mathbf{G} \). Each row corresponding to \(\boldsymbol{v}_{1}, \boldsymbol{v}_{2}, \dots, \boldsymbol{v}_{m} \) in \( \mathbf{G} \) represents the incidence vector of an \((m-1)\)-flat in the Euclidean geometry \( \mathrm{EG}(m, 2) \) (see Example~\ref{ex:flat_example}). Specifically, for each \( i \in \{1, \dots, m\} \), we have:
\begin{align*}
\boldsymbol{v}_{i}(2^m) = 1 \quad \text{and} \quad \boldsymbol{v}_{i}(1) = 0.
\end{align*}
This implies that all such flats pass through the point \( P_{2^m} \) (whose coordinates are all ones) and exclude the origin \( P_1 \) (whose coordinates are all zeros). \textcolor{black}{Crucially, since the point \( P_{2^m} \) lies on all these flats, they are \emph{pairwise non-parallel}.} By Lemma~\ref{lem:FlatIntersection}, the intersection of any two flats is itself a flat. Thus the intersection of any \( \ell \) distinct, pairwise non-parallel \((m-1)\)-flats from this set yields an \((m-\ell)\)-flat \( L \) in \( \mathrm{EG}(m, 2) \). The incidence vector of \( L \), denoted as \(\boldsymbol{v}_{\sigma_1}\boldsymbol{v}_{\sigma_2}\dots\boldsymbol{v}_{\sigma_{\ell}}\), is constructed as the element-wise product of \(\boldsymbol{v}_{\sigma_1}, \boldsymbol{v}_{\sigma_2}, \dots, \boldsymbol{v}_{\sigma_{\ell}} \). Since all flats $\boldsymbol{v}_i, \, i \in [m]$ pass through point $P_{2^m}$, the flat $L$ also passes through this point. Define the \textit{complementary} flat \( T \) to \( L \) with the incidence vector:
\begin{align}
    \label{eq:flat_T}
\boldsymbol{v}_T = \boldsymbol{v}_{\tau_1}\boldsymbol{v}_{\tau_2}\dots\boldsymbol{v}_{\tau_{m-{\ell}}},
\end{align}
where \( \{\tau_1, \dots, \tau_{m-{\ell}}\} = [m] \setminus \{\sigma_1, \dots, \sigma_{\ell}\} = \{\sigma_{{\ell}+1}, \sigma_{{\ell}+2}, \dots, \sigma_{m}\}\). The flat \( T \) has dimension $\ell$ and also passes through the point \( P_{2^m} \). \textcolor{black}{There are \( 2^{m-\ell} \) pairwise disjoint \( \ell \)-flats in \( \mathrm{EG}(m, 2) \), comprising the flat \( T \) and its \( 2^{m-\ell}-1 \) affine translations.} Let:
\[
T_1 = \{\boldsymbol{y} + \boldsymbol{1}_m, \, \text{for all point }\boldsymbol{y} \,\in\, T\}.
\]
Here, \( T_1 \) is a translation of \( T \) that contains the origin \( P_1 \) but not \( P_{2^m} \). Consequently, \( T_1 \) is an \({\ell}\)-dimensional linear subspace $\mathscr{S}$ of \( \mathrm{EG}(m, 2) \). Express the codeword \( \boldsymbol{x} \) as:
\[
\boldsymbol{x} = \boldsymbol{a} \cdot \mathbf{G} = \sum_{\rho = \rho_1\rho_2\dots\rho_h} a_{\rho} \boldsymbol{v}_{\rho_1}\boldsymbol{v}_{\rho_2}\dots\boldsymbol{v}_{\rho_h},
\]
where the sum is over all subsets \( \{\rho_1, \dots, \rho_h\} \) of \( \{1, \dots, m\} \) with \( h \leq r \) (This generalizes~\eqref{eq:message_to_codeword}). The Hadamard product $\boldsymbol{v}_{\rho_1}\boldsymbol{v}_{\rho_2}\dots\boldsymbol{v}_{\rho_h}$ is the incidence vector of some flat $W$ in $\mathrm{EG}(m, 2)$. Now, summing the coordinates of \( \boldsymbol{x} \) over all points \( P_i \in T_1 \) gives:
\begin{align}
\sum\limits_{i\,:\,P_i \,\in\, T_1} x_i &= \sum_{\rho} a_{\rho} \sum_{i\,:\,P_i \,\in\, T_1} \left( \boldsymbol{v}_{\rho_1}\boldsymbol{v}_{\rho_2}\dots\boldsymbol{v}_{\rho_k} \right)_i = \sum_{\rho} a_{\rho} N(T_1, \rho), \label{eq:decoding}
\end{align}
where \( N(T_1, \rho) \) denotes the number of points in the intersection of \( T_1 \) with the flat \( W \). Since \( T_1 \) is defined as a translation of \( T \) by the all-ones vector \( \boldsymbol{1}_m \), its incidence vector \( \boldsymbol{v}_{T_1} \) is obtained by taking the bitwise complement of the incidence vector of \( T \). Consequently, one has from~\eqref{eq:flat_T}:
\[
\boldsymbol{v}_{T_1} = \overline{\boldsymbol{v}}_{T} = \overline{\boldsymbol{v}}_{\tau_1}\overline{\boldsymbol{v}}_{\tau_2}\dots\overline{\boldsymbol{v}}_{\tau_{m-{\ell}}},
\]
where \( \overline{\boldsymbol{v}}_{\tau_i} \) represents the bitwise complement of \( \boldsymbol{v}_{\tau_i} \). Incidence vector of \( T_1 \,\cap\, W \) by Lemma~\ref{lem:FlatIntersection} is:
\begin{align}\label{eq:intersection_incidence}
\boldsymbol{v}_{T_1\,\,\cap\,\,W} = \boldsymbol{v}_{\rho_1}\boldsymbol{v}_{\rho_2}\dots\boldsymbol{v}_{\rho_{h}}\overline{\boldsymbol{v}}_{\tau_1}\overline{\boldsymbol{v}}_{\tau_2}\dots\overline{\boldsymbol{v}}_{\tau_{m-{\ell}}}.
\end{align}

\textbf{Key Observations:}
\begin{enumerate}
    \item Parity of Intersections: All flats of dimension at least one contain an even (power of two) number of points. \textcolor{black}{(True in general for Euclidean geometry.)}
    \item Intersection Dimension: By Lemma~\ref{lem:FlatIntersection}, the intersection \( T_1 \,\cap\, W \) is a flat whose dimension depends on \( h \) relative to \( {\ell} \).
\end{enumerate}

\textbf{Case Analysis:}
\begin{itemize}
    \item When \( h < \ell \): By~\eqref{eq:intersection_incidence}, the incidence vector \( \boldsymbol{v}_{T_1 \,\cap\, W} \) is given by the element-wise product of \( h + m - \ell \) flats, each of dimension \( m - 1 \). Since \( h < \ell \), it follows that \( h + m - \ell \le m-1 \), and hence the intersection \( T_1 \,\cap\, W \) has dimension at least one. Therefore, \( N(T_1, \rho) \) must be even.

    \item When $h = \ell$ and $W = L$: The intersection \( T_1 \,\cap\, W \) consists solely of one point, so \( N(T_1, \rho) = 1 \). (See Lemma~\ref{lem:UniqueIntersection}.)
    
    \item When \( h={\ell} \) and $W \neq L$, then \( N(T_1, \rho) = 0 \) because the intersection conditions become impossible. (See Lemma~\ref{lem:Zero_intersection_2}.)

    \item When \( {\ell} < h \le r \), then \( N(T_1, \rho) = 0 \) because the intersection conditions become impossible. (See Lemma~\ref{lem:ZeroIntersection}.)
\end{itemize}
    
\textbf{Supporting Lemmas:}

\begin{lemma}\label{lem:UniqueIntersection}
If $h=\ell$ and \( W = L \), then \( N(T_1, \rho) = 1 \).
\end{lemma}
\begin{proof}
When \( W = L \), the incidence vector of \( W \) is \( \boldsymbol{v}_{\sigma_1}\boldsymbol{v}_{\sigma_2}\dots\boldsymbol{v}_{\sigma_{\ell}} \). 
Incidence vector of \( T_1 \,\cap\, W \) by Lemma~\ref{lem:FlatIntersection} is:
\[
\boldsymbol{v}_{\sigma_1}\boldsymbol{v}_{\sigma_2}\dots\boldsymbol{v}_{\sigma_{\ell}}\cdot\overline{\boldsymbol{v}}_{\tau_1}\overline{\boldsymbol{v}}_{\tau_2}\dots\overline{\boldsymbol{v}}_{\tau_{m-{\ell}}}.
\]
A point \( P_j \) lies in this intersection if and only if:
\[
\begin{cases}
    \boldsymbol{v}_{\sigma_1}(j) = \boldsymbol{v}_{\sigma_2}(j) = \dots = \boldsymbol{v}_{\sigma_{\ell}}(j) = 1, \\
    \boldsymbol{v}_{\tau_1}(j) = \boldsymbol{v}_{\tau_2}(j) = \dots = \boldsymbol{v}_{\tau_{m-{\ell}}}(j) = 0.
\end{cases}
\]
Since \( \{\tau_1, \dots, \tau_{m-{\ell}}\} = [m] \setminus \{\sigma_1, \dots, \sigma_{\ell}\} \) and the vectors \( \boldsymbol{v}_1, \dots, \boldsymbol{v}_m \) form the set of all length-$m$ binary vectors in \( \mathbb{F}_2^m \), there exists exactly one such point \( P_j \).
\end{proof}

\begin{lemma}\label{lem:Zero_intersection_2}
If \( h={\ell} \) and $W \neq L$, then \( N(T_1, \rho) = 0 \).
\end{lemma}
\begin{proof}
When \( h = \ell \), the incidence vector of the intersection \( T_1 \,\cap\, W \) is, by~\eqref{eq:intersection_incidence},
\[
\boldsymbol{v}_{\rho_1} \boldsymbol{v}_{\rho_2} \dots \boldsymbol{v}_{\rho_{\ell}} \cdot \overline{\boldsymbol{v}}_{\tau_1} \overline{\boldsymbol{v}}_{\tau_2} \dots \overline{\boldsymbol{v}}_{\tau_{m - \ell}}.
\]
Since \( \{\tau_1, \dots, \tau_{m - \ell}\} = [m] \setminus \{\sigma_1, \dots, \sigma_{\ell}\} \), and \( W \neq L \), the defining flats of \( W \) must differ from those of \( L \), whose incidence vector is \( \boldsymbol{v}_{\sigma_1} \dots \boldsymbol{v}_{\sigma_{\ell}} \). Equivalently, the set $\{\rho_1, \dots, \rho_{\ell}\}$ must differ from the set $\{\sigma_1, \dots, \sigma_{\ell}\} = [m] \setminus \{\tau_1, \dots, \tau_{m - \ell}\}$. Therefore, at least one \( \rho_i \) must coincide with some \( \tau_j \). Without loss of generality, assume \( \rho_1 = \tau_1 \).

Then, for any point \( P_j \in T_1 \,\cap\, W \), we must have:
\[
\begin{cases}
\boldsymbol{v}_{\rho_1}(j) = 1, \\
\overline{\boldsymbol{v}}_{\tau_1}(j) = 1.
\end{cases}
\]
But this leads to a contradiction: \( \boldsymbol{v}_{\rho_1}(j) = 1 \) implies \( \boldsymbol{v}_{\tau_1}(j) = 1 \), hence \( \overline{\boldsymbol{v}}_{\tau_1}(j) = 0 \), contradicting the second condition. Thus, no such point \( P_j \) exists, which implies \( T_1 \,\cap\, W = \emptyset \) and hence \( N(T_1, \rho) = 0 \).
\end{proof}

\begin{lemma}\label{lem:ZeroIntersection}
If \( {\ell} < h \le r \), then \( N(T_1, \rho) = 0 \).
\end{lemma}
\begin{proof}
The incidence vector of \( T_1 \,\cap\, W \) is, by~\eqref{eq:intersection_incidence},
\[
\boldsymbol{v}_{T_1\, \,\cap\,\, W} = \boldsymbol{v}_{\rho_1} \boldsymbol{v}_{\rho_2} \dots \boldsymbol{v}_{\rho_h} \cdot \overline{\boldsymbol{v}}_{\tau_1} \overline{\boldsymbol{v}}_{\tau_2} \dots \overline{\boldsymbol{v}}_{\tau_{m - \ell}}.
\]
Since \( h > \ell \), the total number of (possibly dependent) flat vectors in the product is \( h + (m - \ell) > m \). This implies that at least one \( \rho_i \) must coincide with some \( \tau_j \), leading to a contradiction similar to the one seen in the proof of the previous lemma. Hence, no such intersection can exist.
\end{proof}

From~\eqref{eq:decoding}, all terms \( a_{\rho} N(T_1, \rho) \) vanish except when $W = L$ and $\rho = \sigma_{1}\sigma_2\dots\sigma_{\ell}$ (note that we are summing over $\mathbb{F}_2$ so $N(T_1, \rho)$ vanishes whenever it is even). In this case, \( N(T_1, \sigma_{1}\sigma_2\dots\sigma_{\ell}) = 1 \), yielding:
\[
\sum\limits_{i\,:\,P_i \,\in\, T_1} x_i = a_{\sigma^{\ell}}.
\]
Additionally, since \( T_1 \) contains the origin, it is an \( {\ell} \)-dimensional linear subspace $\mathscr{S}$ ($T_1 \equiv \mathscr{S}$) and so \(|\mathscr{S}| = |T_1| = 2^{\ell}, P_1 \in \mathscr{S} \). Thus, if we denote as $S$ the set of indices corresponding to the points in $\mathscr{S}$, constraints in (\ref{min_recovery_set}) are satisfied, concluding our proof.
\end{proof}

\section{Proof of Theorem~\ref{thm:RecoverySetSize}}
\label{app:proof_theorem5}
\begin{proof}
    Let \( S' \) be any recovery set for \( a_{\sigma^{\ell}} \) different from \( S \). By the definition of recovery sets, we have:
    \[
    \sum\limits_{j \,\in\, S} x_j = a_{\sigma^{\ell}} \quad \text{and} \quad \sum\limits_{j \,\in\, S'} x_j = a_{\sigma^{\ell}}.
    \]
    Adding these two equations, we obtain:
    \begin{equation}\label{eq:sum_zero}
        \sum\limits_{j \,\in\, S} x_j + \sum\limits_{j \,\in\, S'} x_j = a_{\sigma^{\ell}} + a_{\sigma^{\ell}} = 0.
    \end{equation}
    Define the set \( S_1 \) as the symmetric difference of \( S \) and \( S' \):
    \[
    S_1 = (S \,\cup\, S') \setminus (S \,\cap\, S').
    \]
    Equation~\eqref{eq:sum_zero} implies:
\[
\sum_{j \in S_1} x_j = 0 - 2\sum_{j \in S \,\cap\, S_1} x_j = 0\ \text{(over $\mathbb{F}_2$)}, \quad \text{for all codewords } \boldsymbol{x} \in \mathcal{C}\ .
\]
Since each codeword component satisfies \( x_j = \boldsymbol{a} \cdot \boldsymbol{c}_j \), where $\boldsymbol{c}_j$ denotes the $j$-th column of $\mathbf{G}$ ($1 \le j \le n$). We obtain
\[
\boldsymbol{a} \cdot \left( \sum_{j \in S_1} \boldsymbol{c}_j \right) = 0 \quad \text{for all message vectors } \boldsymbol{a}.
\]
This holds if and only if \( \sum_{j \in S_1} \boldsymbol{c}_j = \boldsymbol{0}_k \). Equivalently, in matrix form:
\[
\mathbf{G} \cdot \boldsymbol{\chi}(S_1)^\top = \boldsymbol{0}_k,
\]
where \( \boldsymbol{\chi}(S_1) \) is the incidence vector of \( S_1 \), and \( \boldsymbol{0}_k \) is the all-zero vector of length \( k \). Therefore, \( \boldsymbol{\chi}(S_1) \) belongs to the dual code \( \mathcal{C}^{\perp} \) of \( \mathcal{C} = \mathrm{RM}(r, m) \). The dual of the Reed--Muller code \( \mathrm{RM}(r, m) \) is \( \mathrm{RM}(m - r - 1, m) \) when $r\le m-1$, which has a minimum distance of \( 2^{r+1} \)~\cite{Coding:books/MacWilliamsS77}, and is the single codeword $\boldsymbol{0}_n$ when $r = m$. This implies that any non-zero codeword in \( \mathcal{C}^{\perp} \) must have a weight (number of non-zero coordinates) of at least \( 2^{r+1} \). Hence:
    \[
    \text{wt}(\boldsymbol{\chi}(S_1)) \geq 2^{r+1}.
    \]
    Consequently, the size of \( S_1 \) satisfies:
    \[
    |S_1| \geq 2^{r+1}.
    \]
    Since \( S_1 = S \,\cup\, S' \setminus S \,\cap\, S' \), we have:
    \begin{align}\label{eq:AlternativeSize}
        |S'| = |S_1| - |S| + 2|S \,\cap\, S'| \geq |S_1| - |S| \geq 2^{r+1} - |S|.
    \end{align}
    Given that \( |S| = 2^{\ell} \) from Theorem~\ref{thm:RecoverySet}, it follows:
    \[
    |S'| \geq 2^{r+1} - 2^{\ell} > 2^r,
    \]
    whenever \( {\ell} < r \).
    It also follows from~\eqref{eq:AlternativeSize} that
\begin{align}\label{eq:OutsideElements}
    |S' \setminus S| = |S'| - |S \,\cap\, S'| = |S_1| - |S| + |S \,\cap\, S'| \ge |S_1| - |S| \ge 2^{r+1} - 2^{\ell},
\end{align}
i.e., the number of elements in \( S' \) outside of \( S \) is at least \( 2^{r+1} - |S| = 2^{r+1} - 2^{\ell} \). In particular, this implies that any recovery set \( S' \) of size exactly \( 2^{r+1} - 2^{\ell} \) must be disjoint from \( S \).

    \textbf{Special Case When \( {\ell} = r \):}
    
    If \( {\ell} = r \), substituting into the inequality gives:
    \[
    |S'| \geq 2^{r+1} - 2^r = 2^r.
    \]
    Therefore, when \( {\ell} = r \), the recovery set \( S \) has a size of \( 2^r \), which is the minimum possible size for any recovery set of \( a_{\sigma^r} \). Moreover, other recovery sets also attain this minimum size, as established in Theorem~\ref{thm:ReedDecoding}.
\end{proof}

\section{Proof of Theorem~\ref{thm:RecoverySetCount}}
\label{app:proof_theorem6}
\begin{proof}
    We utilize the same notations established in the proofs of Theorems~\ref{thm:RecoverySet} and~\ref{thm:RecoverySetSize}. Let \( \mathbf{e}_i \) denote the standard basis vector such that \( a_{\sigma^{\ell}} = \boldsymbol{a} \cdot \mathbf{e}_i \). This implies:
    \[
    \mathbf{G} \cdot \boldsymbol{\chi}(S)^{\top} = \mathbf{e}_i.
    \]
    
    From Theorem~\ref{thm:RecoverySet}, the point set \( \{P_j \mid j \in S\} \) forms an \( {\ell} \)-dimensional linear subspace \( \mathscr{S} \) in the Euclidean geometry \( \mathrm{EG}(m, 2) \). Now, consider a recovery set \( S' \) for \( a_{\sigma^{\ell}} \) with size \( |S'| = 2^{r+1} - 2^{\ell} \). Our goal is to establish a one-to-one correspondence between each such set \( S' \) and a \((r+1)\)-dimensional linear subspace \( F \) that contains \( \mathscr{S} \).
    
    \textbf{Step 1: Establishing the Correspondence}
    
    From the proof of Theorem~\ref{thm:RecoverySetSize}, we see that \( S' \) must be disjoint from \( S \). Define as $S_1 = S \,\cup\, S'$ the union of \( S \) and \( S' \) whose size is
    \(
    |S_1| = |S| + |S'| = 2^{\ell} + (2^{r+1} - 2^{\ell}) = 2^{r+1}.
    \)
    
Since $\mathbf{G} \cdot \boldsymbol{\chi}(S_1)^{\top} = \mathbf{G} \cdot \boldsymbol{\chi}(S)^{\top} + \mathbf{G} \cdot \boldsymbol{\chi}(S')^{\top} = \mathbf{e}_i + \mathbf{e}_i = \mathbf{0}_k$, the incidence vector \( \boldsymbol{\chi}(S_1) \) is a codeword in the dual code \( \mathcal{C}^{\perp} \) of \( \mathrm{RM}(r, m) \), which is \( \mathrm{RM}(m - r - 1, m) \). Since the weight of \( \boldsymbol{\chi}(S_1) \) is
    \(
    \text{wt}(\boldsymbol{\chi}(S_1)) = 2^{r+1},
    \)
    it follows that \( \boldsymbol{\chi}(S_1) \) is a minimum-weight codeword in \( \mathcal{C}^{\perp} \). By Theorem~\ref{thm:MinWeightRM}, \( \boldsymbol{\chi}(S_1) \) corresponds to an \((r+1)\)-flat \( F \) in \( \mathrm{EG}(m, 2) \).
    
    \textbf{Step 2: \( F \) is an \((r+1)\)-dimensional Linear Subspace}
    
    Since \( 1 \,\in\, S \subset S_1 \), it follows that \( P_1 \,\in\, \mathscr{S} \subset F \) where \( P_1 \) is the origin. Therefore, \( F \) is an \((r+1)\)-dimensional linear subspace that contains the \( {\ell} \)-dimensional linear subspace \( \mathscr{S} \).
    
    \textbf{Step 3: Establishing the Bijection}
    
    Conversely, assume that \( F \) is an \((r+1)\)-dimensional linear subspace of \( \mathrm{EG}(m, 2) \) containing \( \mathscr{S} \), and let \( S_1 \) denotes the point set \( \{P_j \mid j \in F\} \), with \( \boldsymbol{\chi}(S_1) \) its incidence vector. By Theorem~\ref{thm:MinWeightRM}, \( \boldsymbol{\chi}(S_1) \) is a minimum-weight codeword in the dual code \( \mathcal{C}^\perp \) with weight $\text{wt}(\boldsymbol{\chi}(S_1)) = 2^{r+1}$, and thus satisfies:
\[
\mathbf{G} \cdot \boldsymbol{\chi}(S_1)^{\top} = \boldsymbol{0}_k.
\]
Define \( C = F \setminus \mathscr{S} \), and let \( S' \) be the point set \( \{P_j \mid j \in C\} \), with incidence vector \( \boldsymbol{\chi}(S') \). Then the weight of \( \boldsymbol{\chi}(S') \) is:
\[
\text{wt}(\boldsymbol{\chi}(S')) = \text{wt}(\boldsymbol{\chi}(S_1)) - \text{wt}(\boldsymbol{\chi}(\mathscr{S})) = 2^{r+1} - 2^{\ell}.
\]
Moreover, since \( \boldsymbol{\chi}(S_1) = \boldsymbol{\chi}(S') + \boldsymbol{\chi}(\mathscr{S}) \), it follows that:
\[
\mathbf{G} \cdot \boldsymbol{\chi}(S')^{\top} = \mathbf{G} \cdot \boldsymbol{\chi}(S_1)^{\top} - \mathbf{G} \cdot \boldsymbol{\chi}(\mathscr{S})^{\top} = \boldsymbol{0}_k - \mathbf{e}_i = \mathbf{e}_i.
\]
This shows that \( C \) is a valid recovery set for \( a_{\sigma^{\ell}} \) of size \( 2^{r+1} - 2^{\ell} \). Therefore, there is a bijective correspondence between recovery sets \( S' \) of size \( 2^{r+1} - 2^{\ell} \) and \((r+1)\)-dimensional linear subspaces \( F \subseteq \mathrm{EG}(m, 2) \) containing \( \mathscr{S} \).

    \textbf{Step 4: Counting the Recovery Sets}
    
    The number of such \((r+1)\)-dimensional linear subspaces \( F \) that \textcolor{black}{contain the \( {\ell} \)-dimensional linear subspace \( \mathscr{S} \)} is given by the Gaussian binomial coefficient:
    \[
    \gaussbinom{m - {\ell}}{r + 1 - {\ell}}.
    \]
    This coefficient counts the number of ways to choose an \((r+1 - {\ell})\)-dimensional extension of \( \mathscr{S} \) within the remaining \( m - {\ell} \) dimensions~\cite{NetworkCode:book/MarcusONA}.
    
    \textbf{Step 5: Inclusion of Coordinates Outside \( S \)}
    
To demonstrate that each coordinate \( x_j \) with \( j \notin S \) is included in exactly \( \gaussbinom{m - {\ell} - 1}{r - {\ell}} \) of these recovery sets, observe that its corresponding point \( P_j \) is not an element of \( \mathscr{S} \). This implies that adding \( P_j \) to \( \mathscr{S} \) increases the dimension of the linear subspace by one. Indeed:
\begin{align*}
\text{dim}(\mathscr{S} \,\cup\, \{P_j\}) = \text{dim}(\mathscr{S}) + \text{dim}(\text{span}(\{P_j\})) - \text{dim}(\mathscr{S} \,\cap\, \text{span}(\{P_j\})).
\end{align*}
Since \( P_j \notin \mathscr{S} \), the intersection \( \mathscr{S} \,\cap\, \text{span}(\{P_j\}) \) is trivial (i.e., has dimension 0), and \( \text{span}(\{P_j\}) \) is a 1-dimensional linear subspace ($P_j \notin \mathscr{S}$ therefore $P_j$ must be different from the origin $P_1$). Therefore:
\[
\text{dim}(\mathscr{S} \,\cup\, \{P_j\}) = \text{dim}(\mathscr{S}) + 1 ={\ell} + 1.
\]
    
This means that any \((r+1)\)-dimensional linear subspace \( F \) containing both \( \mathscr{S} \) and \( P_j \) must extend \( \mathscr{S} \) by one additional dimension. The number of such \((r+1)\)-dimensional linear subspaces \( F \) is determined by selecting an \((r - {\ell})\)-dimensional subspace from the remaining \( m - {\ell} - 1 \) dimensions (excluding the dimension added by \( P_j \)). The number of ways to do this is given by the Gaussian binomial coefficient:
\[
\gaussbinom{m - {\ell} - 1}{r - {\ell}}.
\]
    
Therefore, each coordinate \( x_j \) not in \( S \) is included in exactly \( \gaussbinom{m - {\ell} - 1}{r - {\ell}} \) recovery sets of size \( 2^{r+1} - 2^{\ell} \).
    
    Combining these observations, we conclude that:
    Number of recovery sets of size $2^{r+1} - 2^{\ell}$ for $a_{\sigma^{\ell}}$ is $\gaussbinom{m - {\ell}}{r + 1 - {\ell}}$, and each coordinate \( x_j, \, j \notin S \) is present in exactly \( \gaussbinom{m - {\ell} - 1}{r - {\ell}} \) such recovery sets.
\end{proof}

\section{Proof of Theorem~\ref{thm:sum_bound_same_order}}
\label{app:proof_theorem8}
\begin{proof}
    We will establish Theorem~\ref{thm:sum_bound_same_order} by first proving an upper bound (converse) on \( \sum\limits_{j = p(\ell)}^{q(\ell)} \lambda_j \) and then demonstrating that this bound is achievable.

    \begin{enumerate}
        \item \textbf{Upper Bound (Converse)} \\
        We again use a capacity argument. Consider a system with \( n = 2^m \) nodes. Theorems~\ref{thm:RecoverySet} and~\ref{thm:RecoverySetSize} show that for each \( j \) such that \( p(\ell) \leq j \leq q(\ell) \), the standard basis vector \( \mathbf{e}_j \) has:
        \begin{itemize}
            \item One unique recovery set \( S \) of size \( 2^{\ell} \). \textcolor{black}{Note that in the special case \( \ell = r \), all recovery sets have the same size \( 2^r \). The set \( R_{j,1} \) is unique in that its associated point set contains the origin \( P_{1} \) and forms an \( r \)-dimensional linear subspace \( \mathscr{S} \). The remaining recovery sets correspond to \( r \)-flats in \( \mathrm{EG}(m,2) \), each obtained as an affine translation of \( \mathscr{S} \) in $\mathrm{EG}(m, 2)$. (See the example in~\eqref{eq:example_RM24}.)}
            \item Additional recovery sets, each of size at least \( 2^{r+1} - 2^{\ell} \).
        \end{itemize}
        \textcolor{black}{Note that in the case when $\ell = r$, there are multiple recovery sets of the same size $2^r$.} From Remark~\ref{restate:recovery_set}, all recovery sets of size \( 2^{\ell} \) for different \( j \) share the first node (corresponding to column \( \boldsymbol{c}_1 \)). Let \( \lambda_{j,1} \) denote the demand allocated to the recovery set \( S \) of size \( 2^{\ell} \) for \( \mathbf{e}_j \).
        
        Since all these recovery sets share the first node, the total demand allocated to this node cannot exceed 1 (i.e., 100\% utilization). Therefore, we have:
        \begin{align}\label{eq:min_allocate}
        \sum\limits_{j = p({\ell})}^{q({\ell})} \lambda_{j,1} \leq 1.
        \end{align}
        
        Let \( B \) represent the total remaining demand for all \( \mathbf{e}_j \):
        \[
        B = \sum\limits_{j = p({\ell})}^{q({\ell})} (\lambda_j - \lambda_{j,1}).
        \]
        
        These remaining demands must be served by recovery sets of size at least \( 2^{r+1} - 2^{\ell} \ge 2^{\ell}\). To minimize the total capacity used, we need all such recovery sets have exactly size \( 2^{r+1} - 2^{\ell} \) (as small as possible), and maximize the amount of demand served by recovery sets of size $2^{\ell}$, i.e.,~\eqref{eq:min_allocate} to hold. Therefore:
        \[
        B \times (2^{r+1} - 2^{\ell}) \leq 2^m - \left(\sum\limits_{j = p({\ell})}^{q({\ell})} \lambda_{j,1}\right)\cdot 2^{\ell} = 2^m - 2^{\ell}.
        \]
        Solving for \( B \):
        \[
        B \leq \dfrac{2^m - 2^{\ell}}{2^{r+1} - 2^{\ell}}.
        \]
        
        Combining the bounds on \( \lambda_{j,1} \) and \( B \), we obtain:
        \[
        \sum\limits_{j = p({\ell})}^{q({\ell})} \lambda_j = \sum\limits_{j = p({\ell})}^{q({\ell})} \lambda_{j,1} + B \leq 1 + \dfrac{2^m - 2^{\ell}}{2^{r+1} - 2^{\ell}}.
        \]
        
        Thus, the upper bound is established:
        \begin{align}
        \label{eq:max_sum_demand}
        \sum\limits_{j = p({\ell})}^{q({\ell})} \lambda_j \leq 1 + \dfrac{2^m - 2^{\ell}}{2^{r+1} - 2^{\ell}}.
        \end{align}
        
        \item \textbf{Achievability} \\
        To demonstrate that the upper bound is achievable, we construct an allocation of demands that attain this bound. For each \( j \) in the range \( p({\ell}) \leq j \leq q({\ell}) \), set:
        \[
            \lambda_j = 1 + \dfrac{2^m - 2^{\ell}}{2^{r+1} - 2^{\ell}} = \lambda_j^{\max}.
        \]
        Specifically, we define the demand vector \( \boldsymbol{\lambda} \) as:
        \begin{align*}
            \boldsymbol{\lambda} = (\lambda_1, \dots, \lambda_j, \dots, \lambda_k) = \left(0, \dots, 0, 1 + \dfrac{2^m - 2^{\ell}}{2^{r+1} - 2^{\ell}}, 0, \dots, 0\right).
        \end{align*}
        This allocation vector \( \boldsymbol{\lambda} \) lies within the service region \( \mathcal{S}(\mathbf{G}) \), as shown in Theorem~\ref{thm:maximal_achievable_simplex}. For this particular allocation, the total sum of demands is:
        \[
            \sum\limits_{i=p({\ell})}^{q({\ell})} \lambda_i = 1 + \dfrac{2^m - 2^{\ell}}{2^{r+1} - 2^{\ell}},
        \]
        which matches the upper bound established in~\eqref{eq:max_sum_demand}. Similarly, for other values of \( j \) in the range \( [p({\ell}), q({\ell})] \), we can construct analogous allocations where only one \( \lambda_j \) is set to its maximum value while the others remain zero.
        
        By the convexity of the service region \( \mathcal{S}(\mathbf{G}) \), as established in Lemma~\ref{lem:convexity}, any linear combination of these allocation points also lies within \( \mathcal{S}(\mathbf{G}) \). Therefore, the upper bound is achievable.
    \end{enumerate}
\end{proof}
\newpage
\bibliography{bibliography}
\bibliographystyle{IEEEtran}


\end{document}